\documentclass{article} [11 pt]

\usepackage{amsmath,amsthm}     
\usepackage{graphicx}     
\usepackage{hyperref} 
\usepackage{url}
\usepackage{amsfonts} 
\usepackage{amssymb}
\usepackage{siunitx}
\usepackage[mathscr]{euscript}
\usepackage{adjustbox}

\usepackage{multicol}
\usepackage{enumerate}
\usepackage[margin=1.0in,top=1.0in, bottom=1.8in]{geometry}

\ifpdf  \fi

\DeclareSymbolFont{msbm}{U}{msb}{m}{n}
\DeclareMathSymbol{\C}{\mathalpha}{msbm}{'103}
\DeclareMathSymbol{\R}{\mathalpha}{msbm}{'122}
\DeclareMathSymbol{\Q}{\mathalpha}{msbm}{'121}
\DeclareMathSymbol{\Z}{\mathalpha}{msbm}{'132}
\DeclareMathSymbol{\N}{\mathalpha}{msbm}{'116}
\DeclareMathSymbol{\K}{\mathalpha}{msbm}{'113}

\newtheorem{definition}{Definition}   
\newtheorem{lemma}{Lemma}
\newtheorem{cor}[lemma]{Corollary}
\newtheorem{theorem}[lemma]{Theorem}

\newcommand{\old}[1]{{}}

\newcommand{\bi}{\begin{itemize}}
\newcommand{\ei}{\end  {itemize}}
\newcommand{\bt}{\begin{tabbing}}
\newcommand{\et}{\end  {tabbing}}
\newcommand{\be}{\begin{enumerate}}
\newcommand{\ee}{\end  {enumerate}}

\newtheoremstyle{obs}
  {\topsep} 
  {0pt} 
  {} 
  {} 
  {\bfseries} 
  {.} 
  {.5em} 
  {} 

\theoremstyle{obs} \newtheorem{observ}{Observation}

\makeatletter
\def\begin@lgo{\begin{minipage}{1in}\begin{tabbing}
        \quad\=\qquad\=\qquad\=\qquad\=\qquad\=\qquad\=\qquad\=\kill}
\def\end@lgo{\end{tabbing}\end{minipage}}

\makeatother

\makeatletter
\@ifundefined{abovecaptionskip}{\newlength\abovecaptionskip}
\long\def\@makecaption#1#2{
   \vskip \abovecaptionskip
   \setbox\@tempboxa\hbox{{\sf\footnotesize \textbf{#1.} #2}}
   \ifdim \wd\@tempboxa >\hsize         
       {\sf\footnotesize \textbf{#1.} #2\par}
     \else                              
       \hbox to\hsize{\hfil\box\@tempboxa\hfil}
   \fi}
\makeatother

\title{From Hall's Marriage Theorem to Boolean Satisfiability and Back}
\author{
Jonathan Lenchner \thanks{IBM T. J. Watson Research Center, Yorktown Heights, NY USA; lenchner@us.ibm.com.}
}                    
\date{}

\begin{document}


\maketitle

\section*{Abstract}
Motivated by the application of Hall's Marriage Theorem in various LP-rounding problems, we introduce a generalization of the classical marriage problem (CMP) that we call the Fractional Marriage Problem. We show that the Fractional Marriage Problem is NP-Complete by reduction from Boolean Satisfiability (SAT). We show that when we view the classical marriage problem (a.k.a. bipartite matching) as a sub-class of SAT we get a new class of polynomial-time satisfiable SAT instances that we call CMP-SAT, different from the classically known polynomial-time satisfiable SAT instances 2-SAT, Horn-SAT and XOR-SAT.

We next turn to the problem of recognizing CMP-SAT instances, first using SAT embeddings, and then using their embeddings within the universe of Fractional Marriage Problems (FMPs). In the process we are led to another generalization of the CMP that we call the Symmetric Marriage Problem, which is polynomial time decidable and leads to a slight enlargement of the CMP-SAT class. We develop a framework for simplifying FMP problems to identify CMP instances that we call Fragment Logic. Finally we give a result that sheds light on how expressive the FMP need be to still be NP-Complete. The result gives a second NP-Complete reduction of the FMP, this time to Tripartite Matching. We conclude with a wide assortment of suggested additional problems.

\section{Introduction}

Amongst the class of facility location problems \cite{Wiki-FacilityLocation:2019, DREZNER:2004} is the so-called $K$-center problem, which asks one to optimally place $K$ centers, or depots, so that the maximum distance any client must travel to get to the closest center is minimized.  This problem can be considered geometrically or on a metric (or even non-metric) graph. On a metric graph, our chief area of interest, there are numerous possibilities, including those where the centers have unbounded capacity \cite{HOCHBAUM:1985}, bounded but uniform capacities \cite{KHULLER:1996, AN:2015}, and bounded but not necessarily uniform capacities \cite{AN:2015, CYGAN:2012}. All versions are known to be NP-Hard \cite{HOCHBAUM:1985} and so approximate solutions are sought. The most common technique used to obtain the approximation results is that of LP-rounding \cite{VAZIRANI:2003}, where one starts with an integer or mixed integer linear program, and relaxes the problem to a fully linear program, solving that in polynomial time to get a lower bound on the optimal solution. From the lower bound obtained by the LP one then ``pushes around'' the possibly fractional amount centers are opened as well as the fractional amount that the centers allocate their capacity to clients to get a fully integral, but approximate solution. The last ingredient of the argument is often the use of Hall's Marriage Theorem (see, e.g., \cite{AN:2015}), which we now state.

\begin{theorem} \label{thm:hmt_classical} \textbf{Hall's Marriage Theorem\cite{HALL:1935}} 
Let $\mathscr{S}$ be a finite collection of subsets of a finite set $X$, possibly with repetition. If for any $\mathscr{W} \subset \mathscr{S}$, $|\mathscr{W}| \leq |\cup_{W \in \mathscr{W}} W|$ then there is an injective function $T:\mathscr{S} \rightarrow X$ such that $T(S) \in S$ for all $S \in \mathscr{S}$. Conversely, given a collection $\mathscr{S}$ of subsets of $X$, if such a function $T$ exists, then for any $\mathscr{W} \subset \mathscr{S}$, $|\mathscr{W}| \leq |\cup_{W \in \mathscr{W}} W|$.
\end{theorem}

The folklore way to think about the Marriage Theorem is that you are given a set $G$ of girls and a set $B$ of boys. Each girl identifies a set of boys that she would willingly marry (some girls possibly having the same preferences). Suppose that each boy is willing to marry any girl that is willing to marry him, or no one at all. Then as long as for any subset, $S$, of girls, the size of the union of the set of boys they would like to marry is as large as $|S|$, then there is a pairing of girls to boys satisfying all of the girls' (and hence everyone's) wishes, such that each girl is paired with a different boy. The converse is obvious; for a pairing of girls to boys to exist there cannot be a subset of girls, the union of whose acceptable pairings is smaller than the number of girls in the subset.
\begin{figure}[h]
\centerline{\scalebox{0.20}{\includegraphics{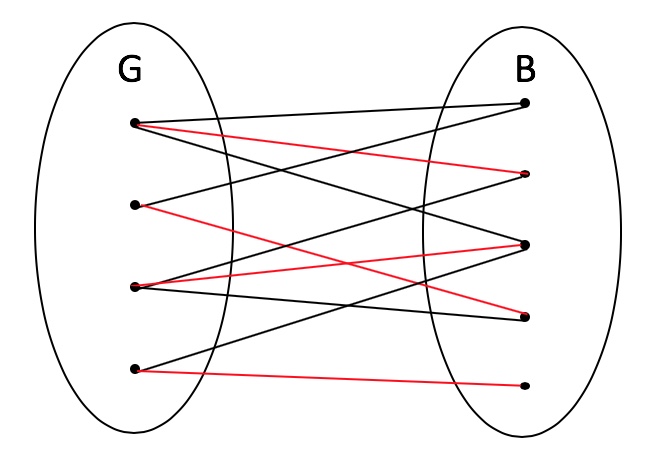}}}
\caption{The folklore view of the Marriage Theorem. The matching in red is one that makes everyone happy.}
\end{figure}

An important thing to note about the Marriage Theorem is that it gives us a cheap ``no certificate,'' thereby establishing that the problem underlying the theorem is in the class co-NP. We will refer to the problem associated with the theorem as the Classical Marriage Problem or CMP, though it is often referred to simply as Bipartite Matching.  Since Bipartite Matching can be viewed as a network flow problem, it can be solved in polynomial time using the Edmonds-Karp implementation \cite{EDMONDS:1972} of the Ford-Fulkerson method \cite{CORMEN:1999}. 

The following generalization of the Classical Marriage Problem is one that naturally presents itself in the context of LP-rounding. In the CMP each girl had a list of boys that she would be willing to marry. In the Fractional Marriage Problem (FMP), each girl can specify not just individual boys, but fractions of boys, as long as the fractions sum up to unity. In other words Susan may be willing to accept either Sam or \{0.5 Harry, 0.5 Frank\} or \{0.25 Joe, 0.25 George, 0.5 Elvis\}.  If each girl specifies such a set of tradeoffs, is it easy to tell if a given set of boys can satisfy at least one tradeoff from each girl? A Hall-like theorem for the FMP would be of substantial help in many LP-rounding problems.

A solution to the FMP is the identification of a collection of boys satisfying the requirements and, moreover, a specification of which tradeoff of each girl is satisfied. \textbf{Henceforth we will assume that the set of boys is just the collection of boys mentioned at least in part by some girl.} The question then is whether this list of boys is sufficient to satisfy all the girls' requirements. The FMP is obviously a generalization of the CMP in that any CMP instance is an FMP instance\footnote{Strictly speaking the CMP requires that each girl be paired with only one boy, while in the FMP we allow for the possibility that a set of boys may satisfy several of the requirements of some of the girls. However, given an assignment of boys to girls in the CMP such that a set of girls are assigned multiple boys, for each such girl we can pick an arbitrary satisfying boy and drop the others to obtain a unique pairing. Hence there is no loss of generality in allowing CMPs that allow for multiple pairings.}.

\medskip

The remainder of this paper is structured as follows: We first establish the NP-Completeness of the FMP, thus dashing our hopes that there is a Hall-like theorem for the FMP unless NP = co-NP. We then establish some normal form reductions for the FMP that are similar to the CNF and 3-CNF reductions that are proved for Boolean formulas. We next show that the CMP maps to SAT yielding a new polynomial-time solvable SAT variant different from the classically known polynomial-time solvable SAT variants 2-SAT, Horn-SAT and XOR-SAT. We call this new family of poly-time solvable SAT instances CMP-SAT. For this family to be a useful one we need to be able to readily identify CMP-SAT instances. We address this issue from two angles: first recognizing CMP-SAT instances from within SAT and second, mapping back to the FMP, recognizing them from within the FMP. For recognition within the FMP we develop something that we call Fragment Logic and show how it works in practice. We provide a slight generalization of the CMP that we call the Symmetrical Marriage Problem that in turn leads to a slightly larger class of poly-time solvable SAT instances and we discuss recognizing instances of this family. We finish up by proving a simple theorem that establishes a limit on the power of Fragment Logic and then conclude with a wide range of suggested additional problems.

\section{NP-Completeness of the FMP}

Henceforth, let us refer to the set of tradeoffs or \textbf{tradeoff set} for a given $g \in G$ by $\mathscr{T}_g$. A sample element $t \in \mathscr{T}_g$ with  $t = \{p_j b_j : b_j \in B,~0 \leq p_j \leq 1,~\sum p_j = 1\}$ will be referred to as a \textbf{tradeoff}. A Fractional Marriage Problem or Fractional Marriage Problem instance will be designated by a triple $(G, B, \{\mathscr{T}_g\}_{g \in G})$.

\begin{theorem} The FMP is NP-Complete \end{theorem}

\begin{proof}
The proof is by reduction from Boolean Satisfiability (SAT). We will consider an arbitrary Boolean formula in conjunctive normal form (CNF) and show how to translate the associated satisfiability problem into an equivalent fractional marriage problem. 

Let us start by considering the simple case where each clause consists entirely of positive literals, an easily satisfiable problem.  
\begin{equation*} \label{eqn:all_positive_literals}
	\psi = (B_1 \vee B_2 \vee B_3) \land (B_1 \vee B_3) \land (B_2 \vee B_3 \vee B_4).
\end{equation*}

We create a corresponding fractional marriage problem $(G, B, \{\mathscr{T}_g\}_{g \in G})$, also readily solvable, incorporating elements $b_1,...b_4$, loosely taking the place of $B_1,...,B_4$, along with a set of auxiliary elements $b'_1,...,b'_7$, as follows:
\begin{flalign*} \label{eqn:all_positive_tradeoffs}
\mathscr{T}_{g_1} =  &~ \{\{\frac{1}{2}b_1, \frac{1}{2}b'_1\}, \{\frac{1}{2}b_2, \frac{1}{2}b'_2\}, \{\frac{1}{3}b_3, \frac{2}{3}b'_3\}\} \notag \\
\mathscr{T}_{g_2} = &~ \{\{\frac{1}{2}b_1, \frac{1}{2}b'_4\}, \{\frac{1}{3}b_3, \frac{2}{3}b'_5\}\}  \\
\mathscr{T}_{g_3} = &~ \{\{\frac{1}{2}b_2, \frac{1}{2}b'_6\}, \{\frac{1}{3}b_3, \frac{2}{3}b'_7\}, \{b_4\}\} \notag
\end{flalign*}
where $G = \{g_1, g_2, g_3\}$ and $B = \{b_1,...,b_4, b'_1,...,b'_7\}$.

In the  translation we replace each clause in $\psi$ with a tradeoff set and each literal with a tradeoff.  Note that for each of the $k_i$ appearances of the variable $B_i$ we have a tradeoff of the form $\{\frac{1}{k_i}b_i, \frac{k_i - 1}{k_i}b'_\alpha\}$, each time with a different $b'_\alpha$. The elements $b_1, b'_1, b'_4$ in the fractional marriage problem formulation collectively play the role of the variable $B_1$ in the SAT problem. Similarly $b_2, b'_2, b'_6$ in the FMP play the role of $B_2$ in the SAT problem, and so on. Note that the elements of $B$ are sufficient to satisfy all the tradeoffs in all the tradeoff sets $\{\mathscr{T}_g\}_{g \in G}$. There is just enough of $b_1,...,b_4$ and more than enough of $b'_1,...,b'_7$.

Next consider the case where the clauses contain negative literals. For example, suppose we add the clause $\neg B_1 \vee \neg B_3$ to our prior formula to get the CNF formula 
\begin{equation*} \label{eqn:negation_clause}
	\psi' = (B_1 \vee B_2 \vee B_3) \land (B_1 \vee B_3) \land (B_2 \vee B_3 \vee B_4) \land (\neg B_1 \vee \neg B_3).
\end{equation*} 
Note that although we chose tradeoffs of the form $\{\frac{1}{k_i}b_i, \frac{k_i - 1}{k_i}b'_\alpha\}$ when we had only positive literals, we actually had some additional flexibility, that we will use to accommodate the negated literals. Specifically, we actually could have chosen tradeoffs of the form $\{\frac{1}{k'_i}b_i, \frac{k'_i - 1}{k'_i}b'_\alpha\}$ for any $k'_i \geq k_i$. With the additional clause added in the formula above we need extra flexibility to accommodate both $\neg B_3$ and $\neg B_1$.  To accommodate $\neg B_1$ we can choose $k'_1 = k_1 + 1$ and so make the following modifications to the first tradeoff in each of $\mathscr{T}_{g_1}$ and $\mathscr{T}_{g_2}$:
\begin{flalign*}
  \mathscr{T}_{g_1}:  &~ \{\frac{1}{2}b_1, \frac{1}{2}b'_1\} \Rightarrow \{\frac{1}{3}b_1, \frac{2}{3}b'_1\} \notag \\
  \mathscr{T}_{g_2}:  &~ \{\frac{1}{2}b_1, \frac{1}{2}b'_4\} \Rightarrow \{\frac{1}{3}b_1, \frac{2}{3}b'_4\}.
\end{flalign*}

With these changes, $\neg B_1$ can be taken to be $\{\frac{1}{2}b'_1, \frac{1}{2}b'_4\}$, with the idea that if this tradeoff is accommodated using $b'_1$ and $b'_4$ then none of the tradeoffs featuring $b_1$ can be accommodated since there will not be enough left of either $b'_1$ or $b'_4$.

For $\neg B_3$ we make the analogous modification to $k_3$ that we previously made to $k_1$, namely, $k'_3 = k_3 + 1$, yielding the following modifications to the respective tradeoffs involving $B_3$:
\begin{flalign*}
  \mathscr{T}_{g_1}:  &~ \{\frac{1}{3}b_3, \frac{2}{3}b'_3\} \Rightarrow \{\frac{1}{4}b_3, \frac{3}{4}b'_3\} \notag\\
  \mathscr{T}_{g_2}:  &~ \{\frac{1}{3}b_3, \frac{2}{3}b'_5\} \Rightarrow \{\frac{1}{4}b_3, \frac{3}{4}b'_5\} \\
  \mathscr{T}_{g_3}:  &~ \{\frac{1}{3}b_3, \frac{2}{3}b'_7\} \Rightarrow \{\frac{1}{4}b_3, \frac{3}{4}b'_7\}. \notag
\end{flalign*}
With these changes to the tradeoffs, we take $\neg B_3$ to be $\{\frac{1}{3}b'_3, \frac{1}{3}b'_5, \frac{1}{3}b'_7\}$, and so get the following Fractional Marriage Problem that is equivalent to the satisfiability problem for $\psi'$:
\begin{flalign*}
	\mathscr{T}_{g_1} = &~ \{\{\frac{1}{3}b_1, \frac{2}{3}b'_1\}, \{\frac{1}{2}b_2, \frac{1}{2}b'_2\}, \{\frac{1}{4}b_3, \frac{3}{4}b'_3\}\}  \\
	\mathscr{T}_{g_2} = &~ \{\{\frac{1}{3}b_1, \frac{2}{3}b'_4\}, \{\frac{1}{4}b_3, \frac{3}{4}b'_5\}\}  \\
	\mathscr{T}_{g_3} = &~ \{\{\frac{1}{2}b_2, \frac{1}{2}b'_6\}, \{\frac{1}{4}b_3, \frac{3}{4}b'_7\}, \{b_4\}\}  \\
	\mathscr{T}_{g_4} = &~ \{\{\frac{1}{2}b'_1, \frac{1}{2}b'_4\}, \{\frac{1}{3}b'_3, \frac{1}{3}b'_5, \frac{1}{3}b'_7\}\}. 
\end{flalign*}

In general we may have $k$ positive instances of a variable and $\ell$ negative instances of the same variable appearing across $k + \ell$ of a total of $N$ clauses. 
If we denote a given variable by $B$ and set $M = \max(k, \ell)$ then a positive instance of $B$ will be replaced by the tradeoff
\begin{equation*}
	\{\frac{1}{M+1}b, \frac{M}{M+1}b'_{\alpha}\},
\end{equation*}
where $b'_{\alpha}$ denotes the next not-yet utilized auxiliary element. We will employ a total of $M$ auxiliary elements across all clauses in which the variable $B$ appears. We replace all instances of $\neg B$ by the tradeoff 
\begin{equation*}
	\{\frac{1}{M}b'_{\beta_1},..., \frac{1}{M}b'_{\beta_M}\},
\end{equation*}
where the elements $b'_{\beta_i}$ include each of the $k$ $b'_{\alpha}$s from the positive literals plus $\ell-k$ additional not-yet used ones if $\ell > k$.

Arguing as we have already, we find that for any given SAT formula in CNF we can produce an equivalent FMP instance in the sense that the SAT formula is satisfiable iff the FMP is satisfiable. It follows that the FMP is NP-hard. Since we can plainly check the validity of a purported FMP solution in polynomial time, the FMP is NP-Complete.
\end{proof}

\medskip

Alas it follows that there is likely no analog of Hall's Marriage Theorem for the FMP since such a theorem would provide a ``no certificate'' for FMP instances, meaning that the FMP would be both NP-Complete and in co-NP, which is not possible unless NP = co-NP.

\section{Normal Form Theorems for the FMP}

Just like an arbitrary SAT formula can be reduced to Conjunctive Normal Form (CNF), and further with at most three literals per clause, we can do something similar for FMP instances. Proof of these theorems are contained in Appendix \ref{app:normal_form_proofs}.

\begin{theorem} \label{thm:fmp_1st_normal_form} \textbf{First Normal Form Reduction} 
Every Fractional Marriage Problem $\mathscr{F} = (G, B, \{\mathscr{T}_g\}_{g \in G})$ has an equivalent formulation $\mathscr{F}' = (G', B', \{\mathscr{T}_g\}_{g \in G'})$ with at most three tradeoffs per tradeoff set $\{\mathscr{T}_g\}_{g \in G'}$. In other words, $\mathscr{F}$ is solvable by $B$ if and only if $\mathscr{F}'$ is solvable by $B'$. 

Moreover, at most one tradeoff in any tradeoff set need be fractional, and within such tradeoff sets the fractional tradeoffs can be constrained to appear in tradeoff sets all of which comprise just two tradeoffs. 

The reduction from $\mathscr{F}$ to $\mathscr{F}'$ can be performed in polynomial time.
\end{theorem}

\begin{definition} We shall refer to an element $pb$ with $0 \leq p \leq 1$ that is part of a tradeoff $t$, which in turn is an element of a tradeoff set $\mathscr{T}_g$, either as a \textbf{fractional requirement} or as a \textbf{fragment}. Given a fractional marriage problem $(G, B, \{\mathscr{T}_g\}_{g \in G})$, we let $F$ denote the set of all fragments (fractional requirements) appearing across all tradeoffs and tradeoff sets.
\end{definition}

\begin{definition} Given an FMP $\mathscr{F} = (G, B, \{\mathscr{T}_g\}_{g \in G})$, an element $b \in B$ is said to be \textbf{free}, or \textbf{free in $\mathscr{F}$}, if $\sum_{pb \in F} p \leq 1$. 
Say that a fragment $pb$ is \textbf{free} if the associated element $b$ is free.
\end{definition}

\begin{theorem} \textbf{Second Normal Form Reduction} \label{thm:fmp_2nd_normal_form}
In the conclusion of the First Normal Form Theorem, we may further assume that each fractional tradeoff has just two elements, in other words each fractional tradeoff is of the form $t = \{p b_i, (1-p) b_j\}$ and $0 < p < 1$. In addition, the elements appearing in the second fractional part, the elements $b_j$, can all be taken to be free.

This further reduction can also be performed in polynomial time.\\
\end{theorem}

The normal form reductions are used in the proof of Theorem \ref{thm:fragment_reduction_is_polynomial}.

\section{CMP $\hookrightarrow$ SAT: A New Family of Poly-Time Solvable SAT Instances}

Although SAT is in general NP-Complete there are three famous SAT sub-families that are solvable in polynomial time: 2-SAT, Horn-SAT and XOR-SAT \cite{WIKI-SATA:2019}. In this section we show that the CMP maps to SAT to give a family of polynomial-time solvable SAT instances that is different from these three.

2-SAT is the collection of Boolean formulas that can be written in CNF with at most two disjuncts per clause.  Note that the disjunct $A \vee B$ is equivalent to $\neg A \rightarrow B$ and $\neg B \rightarrow A$. 
Given a 2-SAT formula $\phi$ we may then form a directed graph $G$ where the set of vertices is the set of variables in $\phi$, together with their negations. For every clause of the form $A \vee B$ we add a directed edge from $\neg A$ to $B$ and from $\neg B$ to $A$. Given such a graph, it is not difficult to see that $\phi$ is satisfiable iff $G$ has no directed cycle containing both a variable and its negation, something that can readily be checked in polynomial time.

Horn-SAT is the collection of Boolean formulas that can be written in CNF with at most one positive literal per clause. There is a simple polynomial-time recursive algorithm for determining if a Horn-SAT formula $\phi$ is satisfiable, and if so, what a satisfying assignment is \cite{WIKI-Horn:2019, DOWLING:1984}. Since there is complete duality between a variable and its negation, if we have a formula $\phi$ and can flip the polarity of a subset of the variables to get a Horn formula we can then just apply the same procedure. It turns out that determining whether a formula is Horn-rewritable in this sense is actually solvable in linear time \cite{HEBRARD:1994}.

XOR-SAT is the collection of Boolean formulas that can be written in CNF-like form but with XOR ($\oplus$) replacing OR ($\vee$) in all clauses. For example: $(A \oplus B \oplus C) \land (\neg A \oplus C) \land (B \oplus \neg C \oplus \neg D)$.
The satisfaction of this formula is identical to the satisfaction of the following system of linear equations modulo $2$:
\begin{flalign*}
	a + b + c = 1 &\\
	(1-a) + c = 1 & \leftrightarrow a + c = 0 \\
	b + (1-c) + (1-d) = 1 & \leftrightarrow b + c + d = 1.
\end{flalign*}
Thus, satisfiability can be determined using Gaussian Elimination over $\mathbb{F}_2$, which has a run-time that is cubic in max$(n,m)$, where $n =$ number of variables, $m =$ number of equations.

\medskip

Let's return now to the Classical Marriage Problem. It is easy to encode a CMP instance as a SAT formula.  Consider the following CMP:
\begin{flalign} \label{sample_cmp}
	\mathscr{T}_{g_1} & = \{b_1, b_2, b_3\} \notag \\
	\mathscr{T}_{g_2} & = \{b_1, b_3\} \notag \\
	\mathscr{T}_{g_3} & = \{b_2, b_4\}  \\
	\mathscr{T}_{g_4} & = \{b_1, b_4\}. \notag 
\end{flalign}

For each element $b_i$ appearing in the tradeoffs we create a set of variables $\{b_{ij}\}$, where $b_{ij}$ represents the element $b_i$ in the $j$th tradeoff set.  Since $b_1$ can be used to satisfy exactly one of the tradeoffs $\mathscr{T}_{g_i}$ above, we replace $b_1$ in $\mathscr{T}_{g_1}$ by $b_{11} \land \neg b_{12} \land \neg b_{14}$, replace $b_1$ in $\mathscr{T}_{g_2}$ by $\neg b_{11} \land b_{12} \land \neg b_{14}$, and so on.

In place of the above tradeoff sets we therefore get the following sub-formulas: 
\begin{flalign} \label{intial_clauses}
	\phi_1 =~ & (b_{11} \land \neg b_{12} \land \neg b_{14}) \vee (b_{21} \land \neg b_{23}) \vee (b_{31} \land \neg b_{32}) \notag \\
	\phi_2 =~ & (\neg b_{11} \land b_{12} \land \neg b_{14}) \vee(\neg b_{31} \land b_{32}) \notag \\
	\phi_3 =~ & (\neg b_{21} \land b_{23}) \vee (b_{43} \land \neg b_{44}) \\
	\phi_4 =~ & \notag (\neg b_{11} \land \neg b_{12} \land b_{14}) \vee (\neg b_{43} \land b_{44}).
\end{flalign}
And the full SAT formula that we need to satisfy is $\phi = \phi_1 \land \phi_2 \land \phi_3 \land \phi_4$.

With some additional work we get the following 3-CNF reduction:
\begin{flalign} \label{eqn:messy_3cnf}
	\phi_1 =~ & (((b_{11} \land \neg b_{12} \land \neg b_{14}) \vee b_{21}) \land ((b_{11} \land \neg b_{12} \land \neg b_{14}) \vee \neg b_{23})) \vee (b_{31} \land \neg b_{32}) \notag \\
	=~ &(((b_{11} \vee b_{21}) \land (\neg b_{12} \vee b_{21}) \land (\neg b_{14} \vee b_{21})) \land ((b_{11} \vee \neg b_{23}) \land (\neg b_{12} \vee  \neg b_{23}) \land (\neg b_{14} \vee \neg b_{23}))) \vee (b_{31} \land \neg b_{32}) \notag \\
	=~ &((b_{11} \vee b_{21}) \land (\neg b_{12} \vee b_{21}) \land (\neg b_{14} \vee b_{21}) \land (b_{11} \vee \neg b_{23}) \land (\neg b_{12} \vee  \neg b_{23}) \land (\neg b_{14} \vee \neg b_{23})) \vee (b_{31} \land \neg b_{32}) \notag \\
	=~ &(((b_{11} \vee b_{21}) \land (\neg b_{12} \vee b_{21}) \land (\neg b_{14} \vee b_{21}) \land (b_{11} \vee \neg b_{23}) \land (\neg b_{12} \vee  \neg b_{23}) \land (\neg b_{14} \vee \neg b_{23})) \vee b_{31}) \land \notag \\
	&(((b_{11} \vee b_{21}) \land (\neg b_{12} \vee b_{21}) \land (\neg b_{14} \vee b_{21}) \land (b_{11} \vee \neg b_{23}) \land (\neg b_{12} \vee  \neg b_{23}) \notag \land (\neg b_{14} \vee \neg b_{23})) \vee \neg b_{32}) \notag \\
	=~ &(b_{11} \vee b_{21} \vee b_{31}) \land (\neg b_{12} \vee b_{21} \vee b_{31}) \land (\neg b_{14} \vee b_{21} \vee b_{31}) \land (b_{11} \vee \neg b_{23} \vee b_{31}) \land (\neg b_{12} \vee  \neg b_{23} \vee b_{31})  \notag \\
	& \land (\neg b_{14} \vee \neg b_{23} \vee b_{31}) \land (b_{11} \vee b_{21} \vee \neg b_{32}) \land (\neg b_{12} \vee b_{21} \vee \neg b_{32}) \land (\neg b_{14} \vee b_{21} \vee \neg b_{32}) \notag \\
	& \land (b_{11} \vee \neg b_{23} \vee \neg b_{32}) \land (\neg b_{12} \vee  \neg b_{23} \vee \neg b_{32}) \land (\neg b_{14} \vee \neg b_{23} \vee \neg b_{32})  \\
	\phi_2 =~ & ((\neg b_{11} \land b_{12} \land \neg b_{14}) \vee \neg b_{31}) \land ((\neg b_{11} \land b_{12} \land \neg b_{14}) \vee b_{32}) \notag \\
	=~ & (\neg b_{11} \vee \neg b_{31}) \land (b_{12} \vee \neg b_{31}) \land (\neg b_{14} \vee \neg b_{31}) \land (\neg b_{11} \vee b_{32}) \land (b_{12} \vee b_{32}) \land (\neg b_{14} \vee b_{32}) \notag \\
	\phi_3 =~ & ((\neg b_{21} \land b_{23}) \vee b_{43}) \land ((\neg b_{21} \land b_{23}) \vee \neg b_{44}) \notag \\
	=~ & (\neg b_{21} \vee b_{43}) \land (b_{23} \vee b_{43}) \land (\neg b_{21} \vee \neg b_{44}) \land (b_{23} \vee \neg b_{44}) \notag \\
	\phi_4 =~ & ((\neg b_{11} \land \neg b_{12} \land b_{14}) \vee \neg b_{43}) \land ((\neg b_{11} \land \neg b_{12} \land b_{14}) \vee b_{44}) \notag \\
	=~ & (\neg b_{11} \vee \neg b_{43}) \land (\neg b_{12} \vee \neg b_{43}) \land (b_{14} \vee \neg b_{43}) \land (\neg b_{11} \vee b_{44}) \land (\neg b_{12} \vee b_{44}) \land (b_{14} \vee b_{44}) \notag
\end{flalign}

We now show that $\phi$ is not an example of 2-SAT, Horn-SAT or XOR-SAT. On the one hand, $\phi$ evidently has clauses of length $3$ so is not an example of 2-SAT. Additionally, there are no simple reductions that would allow us to conceivably reduce clauses of length $3$ to clauses of length $2$. In other words (1) there are no variables that appear purely in positive form or purely in negative form across all clauses, since such clauses could just be removed, and (2) there are no conjuncts either of the generic form $(A \vee B) \land A$
or of the generic form $(A \vee B) \land (A \vee \neg B)$, either of which would simplify to just $A$.

On the one hand, the example we have just considered provides a CMP instance that is not a 2-SAT instance since, modulo a detail that we shall get to in a few parapgrahps, there are clauses in the 3-SAT reduction that are of length $3$. However, there is also the much simpler CMP instance $\mathscr{T}_g  = \{b_1, b_2, b_3\}$, which is equivalent to the CNF formula $\phi = b_1 \vee b_2 \vee b_3$, which cannot be represented as a conjunction of clauses of length $2$, even if we add additional variables.

Next, let us rule out Horn-SAT. Looking at the full 3-CNF reduction above for $\phi = \phi_1 \land \phi_2 \land \phi_3 \land \phi_4$, we see that there are clauses with more than one positive literal and so $\phi$ is not Horn. However, we must also show that $\phi$ is not Horn-rewritable -- an argument that we provide in Appendix \ref{app:non_horn_rewritability}.

XOR-SAT is similarly easy to rule out. $A \oplus B$ is equivalent to $(A \vee B) \land (\neg A \vee \neg B)$, while $A \oplus B \oplus C$ is equivalent to $(A \vee B \vee C) \land (\neg A \vee \neg B \vee C) \land (\neg A \vee B \vee \neg C) \land (A \vee \neg B \vee \neg C)$ and in the 3-CNF representation of $\phi$ we have no two clauses that even have the same variables. Hence:

\begin{theorem} \label{thm:no_union} The family of CMP instances, when rendered as SAT formulas, is not contained in 2-SAT $\cup$ Horn-SAT $\cup$ XOR-SAT.
\end{theorem}

For this family, which from now on we shall refer to as CMP-SAT, to form a useful polynomially-decidable SAT family we must be able to recognize instances of CMP-SAT.  We take this matter up in the next section.

\section{Recognizing CMP-SAT Instances}

It turns out that there is an alternative and useful way to represent CMP instances. Let us introduce a new Boolean operator, $\oplus_{\leq 1}(A_1,...,A_N)$ which is TRUE if and only if at most one of $A_1,...,A_N$ is TRUE.  Returning to the CMP instance that we previously considered, we have the following equivalence:
\begin{flalign} \label{eqn:cmp}
	& \mathscr{T}_{g_1} = \{b_1, b_2, b_3\} \notag \\
	& \mathscr{T}_{g_2} = \{b_1, b_3\} \notag \\
	& \mathscr{T}_{g_3} = \{b_2, b_4\}  \\
	& \mathscr{T}_{g_4} = \{b_1, b_4\}. \notag 
\end{flalign}
\begin{flalign} \label{eqn:canonical_embedding}
	& b_{11} \vee b_{21} \vee b_{31} \notag \\
	& b_{12} \vee b_{32}  \notag \\ 
	& b_{23} \vee b_{43} \notag \\
	& b_{14} \vee b_{44} \notag \\
	&\oplus_{\leq 1}(b_{11}, b_{12}, b_{14}) \\
	& \oplus_{\leq 1}(b_{21}, b_{23}) \notag \\
	& \oplus_{\leq 1}(b_{31}, b_{32}) \notag \\
	& \oplus_{\leq 1}(b_{43}, b_{44}). \notag 
\end{flalign}

Saying that at most one of a set of variables is TRUE is the same as saying that for every pair of the variables in the set, at least one variable in the pair is FALSE. Hence we have: 
\begin{flalign*}
	& \oplus_{\leq 1}(A_1,...,A_N) \leftrightarrow \bigwedge_{1 \leq i < j \leq N} (\neg A_i \vee \neg A_j) . 
\end{flalign*}
Thus, (\ref{eqn:canonical_embedding}) is equivalent to\footnote{Theorem \ref{thm:no_union} can equally well be proved using this CNF encoding.}:
\begin{flalign} \label{eqn:canonical_embedding2}
	& b_{11} \vee b_{21} \vee b_{31} \notag\\
	& b_{12} \vee b_{32} \notag\\ 
	& b_{23} \vee b_{43} \notag\\
	& b_{14} \vee b_{44} \notag\\
	& (\neg b_{11} \vee \neg b_{12}) \land (\neg b_{11} \vee \neg b_{14}) \land (\neg b_{12} \vee \neg b_{14}) \\
	& \neg b_{21} \vee \neg b_{23} \notag\\
	& \neg b_{31} \vee \neg b_{32} \notag\\
	& \neg b_{43} \vee \neg b_{44}. \notag
\end{flalign}

There is, however, a certain non-uniqueness to this encoding.  Since we are encoding four tradeoff sets, we must have at least four elements among all the $b_{ij}$ that are TRUE. It follows that the at-most-one-of relations, $\oplus_{\leq 1}$, could just as well have been encoded as ``exactly-one-of''. Thus we get a truth-table-wise equivalent formula, by adding any number of the following clauses:
\begin{flalign*} 
	& b_{11} \vee b_{12} \vee b_{14} \\
	& b_{21} \vee b_{23} \\
	& b_{31} \vee b_{32} \\
	& b_{43} \vee b_{44}.
\end{flalign*}
We shall have more to say about this matter (see, in particular, Observation \ref{obs3} and theorems \ref{thm:smp} and \ref{thm:sum_up}).

To summarize: every CMP is  a collection of OR clauses with all positive literals, and $\oplus_{\leq 1}$ statements, where a $\oplus_{\leq 1}$ statement with $K$ arguments is the disjunction of all $K \choose{2}$ pairs of the associated negated literals. Since the OR clauses with all positive literals can be of arbitrary length it is worth recalling how such disjunctions are reduced to equi-satisfiable 3-CNF formulas:
\begin{flalign*}
	& A \vee B \vee C \vee D \vee E \leftrightarrow (A \vee B \vee X) \land (\neg X \vee C \vee Y) \land (\neg Y \vee D \vee E)
\end{flalign*}
$A \vee B \vee C \vee D \vee E$ is TRUE iff there is a truth value assignment of $X, Y$ making the right hand side above TRUE. $X$ and $Y$ are sometimes called ``bridge'' variables and should not appear in any other clauses.

\smallskip

Thus, to \textbf{recognize} a CMP instance of 3-CNF SAT we must observe the following:
\begin{enumerate}[(1)]
	\item \label{item:1st_step}There is some number of $\oplus_{\leq 1}$ statements\footnote{Possibly none, which would correspond to a CMP in which each girl picks a set of boys that is disjoint from the set picked by each other girl.}, recognizable because when these statements are expanded, if we take the set of constituent variables, every two members of the set appear negated and together by themselves in a clause. No variable appears more than once across all the $\oplus_{\leq 1}$ statements.
	\item \label{item:key_step} There is some strictly positive number of clauses with all positive literals, possibly containing bridge variables. No variable appears more than once across all these clauses. There are optionally, some number of clauses, each of which contains all positive literals (again possibly containing bridge variables) the variables of which are identical to the variables appearing in some $\oplus_{\leq 1}$ statement.
	\item There are no other clauses.
	\item \label{item:generalization_step} If we strike out the optional clauses for a moment, then every variable appearing in negated form in one of the expanded $\oplus_{\leq 1}$ statements appears exactly once in positive form in a clause with all positive literals. Moreover, no two variables from the same expanded $\oplus_{\leq 1}$ statement appear in the same all-positive clause.\\
\end{enumerate}
And to \textbf{solve} such an instance:
\begin{enumerate}[(1)]
	\item Turn any all-positive clauses that include bridge variables into disjunctions of arbitrary length of all positive literals without bridge variables.
	\item Turn ${L}\choose{2}$ sets of binary negated literals into $\oplus_{\leq 1}$ statements.  Suppose there are $K$ of these $\oplus_{\leq 1}$ statements.
	\item Turn each all-positive clause $C_i$ into an integral tradeoff set $\mathscr{T}_i$ with the same named elements.
	\item Identify each $\oplus_{\leq 1}$ statement with a representative unique element name, say $B_j{:}~1 \leq j \leq K$.
	\item In the tradeoff sets created in step (3) replace each variable $V_{\alpha}$ appearing in each $\oplus_{\leq 1}$ statement with its unique representative. Keep track of all of these mappings (i.e. mappings of the form $V_{\alpha} \hookrightarrow B_j$).
	\item The collection of tradeoff sets now defines a CMP that we can solve using known methods.
	\item If there is no solution to the CMP then there is no solution to the associated SAT problem.  On the other hand, let $T()$ be a transversal to the collection $\{\mathscr{T}_i\}$ as defined in Theorem \ref{thm:hmt_classical}. In other words, $T(\mathscr{T}_i) \in \mathscr{T}_i, \forall i$, and $T(\mathscr{T}_i) \neq T(\mathscr{T}_j)$ as long as $i \neq j$. Because of the correspondence between tradeoff sets and clauses, $T$ picks a unique element of each $C_i$, and mapping backwards (if necessary) using the inverse of the recorded transformations $V_{\alpha} \hookrightarrow B_j$, we get a set of variables $\{V_{\alpha}\}$, which, when set to TRUE, coupled with setting all other non-bridge variables to FALSE (to satisfy the ``at most one of'' nature of the $\oplus_{\leq 1}$ clauses), satisfies all clauses in the associated SAT formula.
	\item These variable settings guarantee the satisfaction of all of the all-positive disjuncts and $\oplus_{\leq 1}$ clauses. For long all-positive disjuncts we still need to tweak the bridge variables. For example, in the formula
	\begin{flalign*}
		& (A \vee B \vee X) \land (\neg X \vee C \vee Y) \land (\neg Y \vee D \vee E)
	\end{flalign*}
we will be guaranteed that one of $\{A, B, C, D, E\}$ is TRUE, but still will need to set $X$ and $Y$ appropriately to guarantee the truth of each clause.
\end{enumerate}

We provide an example illustrating this recognition and solution process in Appendix \ref{app:recognizing_and_solving}.

While we have written down conditions for recognizing CMP instances as we have chosen to encode them, let us explore whether these can be generalized somewhat to either still give CMP instances, or at least give poly-time solvable SAT/FMP instances. Here are the key elements from the recognition steps (\ref{item:1st_step}) - (\ref{item:generalization_step}) that we would like to generalize:

\begin{enumerate}
\item[(\ref{item:key_step})] ....~There are optionally, some number of clauses, each of which contains all positive literals (again possibly containing bridge variables) the variables of which are identical to the variables appearing in some $\oplus_{\leq 1}$ statement.
\item[(\ref{item:generalization_step})]  If we strike out the optional clauses for a moment, then every variable appearing in negated form in one of the expanded $\oplus_{\leq 1}$ statements appears exactly once in positive form in a clause with all positive literals. Moreover, no two variables from the same expanded $\oplus_{\leq 1}$ statement appear in the same all-positive clause.
\end{enumerate}

\medskip

\begin{observ} \label{obs1} Variables appearing in the various $\oplus_{\leq 1}$ statements need not appear in any $\bigvee$ clause. \end{observ}
\noindent \textbf{Reason:} We may remove any such variables from their associated $\oplus_{\leq 1}$ statements and get an equi-satisfiable formula.

\begin{observ} \label{obs2} Variables appearing in the same $\oplus_{\leq 1}$ statement can appear together in a single $\bigvee$ clause rather than in distinct $\bigvee$ clauses.\end{observ}
\noindent \textbf{Reason:} Suppose we have $\oplus_{\leq 1}(A_1,...,A_K)$, together with $A_i$ and $A_j$ appearing in the same $\bigvee$ clause for some $1 \leq i < j \leq k$. Given that $A_i$ and $A_j$ appear in no other clauses, we can just remove $A_j$, say, from the $\bigvee$ clause and get an equi-satisfiable formula.

\medskip

In order to obtain our final observation, which refines the statement about the optional $\bigvee$ clauses in recognition step (\ref{item:key_step}), let us pause to establish the following:

\begin{theorem} \label{thm:smp} \textbf{Symmetric Marriage Theorem} Let $G$ be a set of Girls and $B$ a set of Boys. Suppose a subset of the girls each makes a list of the boys she is willing to marry and similarly a subset of the boys each makes a list of the girls he is willing to marry. Assume that any girl who does not make a list is willing to marry any boy, or not be married at all, and analogously for the boys.  If there is a pairing of girls with boys that makes everyone happy, then such a pairing can be found in polynomial time. Further, if no such pairing exists, this fact too can be determined in polynomial time.
\end{theorem}

\smallskip

Let us call the problem associated with the Symmetric Marriage Theorem the \textbf{Symmetric Marriage Problem} or \textbf{SMP}. The CMP is a special case of the SMP where each girl makes a list but no boy does. Consideration of the SMP will be central to our final observation about generalizing the recognition criteria. Given an SMP instance, let $G_L \subset G$ denote the girls with lists and $B_L \subset B$ denote the boys with lists.  With that notation in hand, denote an SMP instance using the tuple $\mathscr{S} = (G, G_L, B, B_L, \{\mathscr{T}_g\}_{g \in G_L}, \{\mathscr{T}_b\}_{b \in B_L})$, where of course $\{\mathscr{T}_b\}_{b \in B_L}$ denotes the preferences of the boys with lists and $\{\mathscr{T}_g\}_{g \in G_L}$ the preferences of the girls with lists.


\begin{proof} 
Given an SMP instance $\mathscr{S} = (G, G_L, B, B_L, \{\mathscr{T}_g\}_{g \in G_L}, \{\mathscr{T}_b\}_{b \in B_L})$, create a graph consisting of four sets of nodes. Two of the sets will be the elements of $G$ and $B$, which we will denote by these same letters. In addition, create a node $L_g$ for each girl $g \in G_L$, with $L_g$ representing the set of \textit{boys} on girl $g$'s list for whom $g$ is also on the respective boy's list. In other words $b \in L_g$ iff $b$ is on $g$'s list and $g$ is on $b$'s list. If $g$ and $b$ are both on each other's lists let us call them \textbf{list-compatible}. Denote the set $\{L_g\}_{g \in G_L}$ by $L_G$. The fourth group of nodes then consists of a node $L_b$ for each boy $b \in B_L$, with $L_b$ representing the set of \textit{girls} on boy $b$'s list that are list-compatible with him. Denote the set $\{L_b\}_{b \in B_L}$ by $L_B$.

From the SMP $\mathscr{S}$, and the four sets of nodes as just described, we form a graph, which we shall denote by $G^*(\mathscr{S})$, where we connect a vertex $g \in G$ to vertex $b \in B$ iff either (i) $b$ is on $g$'s list and $b$ has no list, or (ii) $g$ is on $b$'s list and $g$ has no list. Next if $g \in G_L$ and $b \in B_L$ and $g$ and $b$ are list-compatible we connect $g$ to $L_b$ and $b$ to $L_g$. A sample SMP instance $\mathscr{S}$ and associated graph $G^*(\mathscr{S})$ are depicted in Figure \ref{fig:4_partitions}.
\begin{figure}[h] 
\centerline{\scalebox{0.33}{\includegraphics{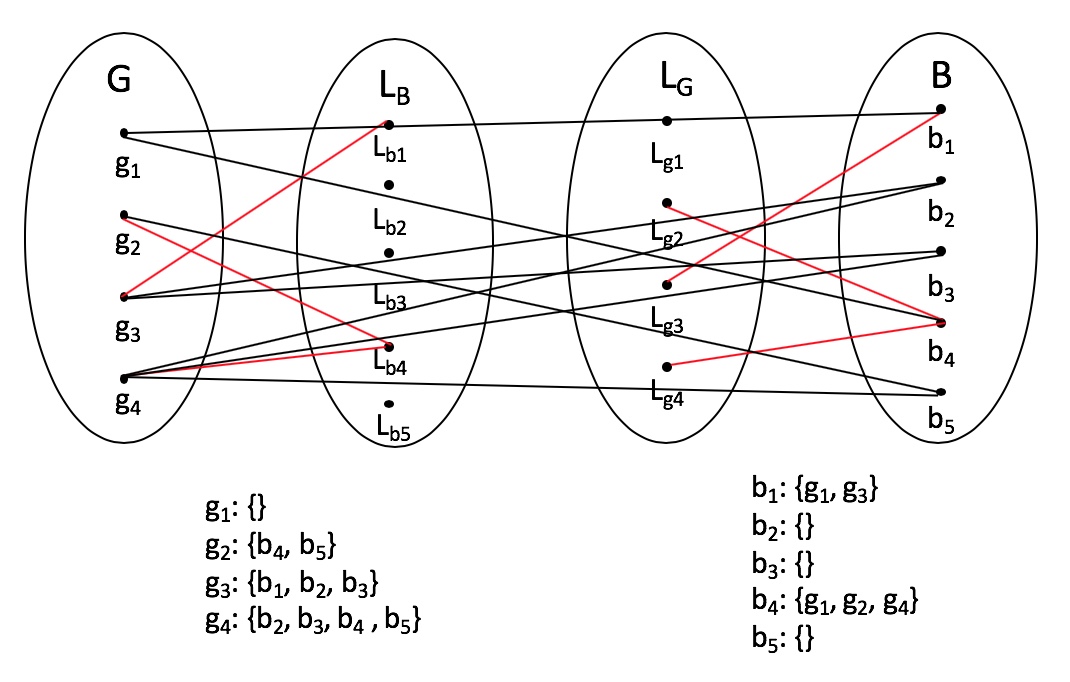}}}
\caption{A sample SMP instance $\mathscr{S}$ together with associated graph $G^*(\mathscr{S})$. The lists of the girls are given on the bottom left and the lists of the boys are given on the bottom right. Edges in $G^*(\mathscr{S})$ between vertices in $G$ and vertices in $B$ are drawn in black, while edges incorporating vertices from $L_B$ or $L_G$ are drawn in red for clarity.} \label{fig:4_partitions}
\end{figure} 

We claim that an SMP $\mathscr{S}$ has a solution iff a maximum matching in $G^*(\mathscr{S})$ is of size $|G_L| + |B_L|$. Since a maximum matching can be found in polynomial time for an arbitrary graph using Edmonds' blossum algorithm\footnote{In fact, we can form one partition out of the vertices in $G \cup L_G$ and another out of the vertices in $B \cup L_B$ and get a bipartite graph and solve the problem using bipartite matching.}\cite{EDMONDS:1965}, once we establish this claim we have proven the theorem. 

To see that a solvable SMP $\mathscr{S}$ gives rise to maximum matching $M$ of size $|G_L| + |B_L|$ in $G^*(\mathscr{S})$, let $P$ be the pairing of elements of $G$ with elements of $B$ providing a solution to $\mathscr{S}$. In other words $P$ is a partial function from $G$ to $B$ with $G_L$ a subset of its domain and $B_L$ a subset of its range. If $P(g) = b$ and exactly one of $g$ and $b$ have lists then we place $(g, b)$ in our matching $M$. On the other hand, if $P(g) = b$ and $g$ and $b$ are list-compatible (the only other possibility), then we place $(g, L_b)$ and $(L_g, b)$ in $M$. Since an SMP solution pairs every girl or boy with a list and does not pair anyone twice, $M$ is a matching and of size $|G_L| + |B_L|$. Since all the edges of $G^*(\mathscr{S})$ touch some vertex in either $G_L$ or $B_L$, but no two edges in a matching can include the same vertex, a matching in $G^*$ is necessarily of size at most $|G_L| + |B_L|$. Hence $M$ is a maximum matching.

Now suppose we have a matching $M$ of size $|G_L| + |B_L|$ in $G^*(\mathscr{S})$. We will describe how to turn $M$ into a solution to $\mathscr{S}$. If $M$ contains edges $(g,b)$ and any time $M$ contains an edge $(g_i, L_{b_j})$ it contains the corresponding edge $(L_{g_i}, b_j)$ then we can convert such a matching into a pairing $P$ such that $P(g) = b$ if $(g,b) \in M$ and $P(g_i) = b_j$ if both $(g_i, L_{b_j}), (L_{g_i}, b_j) \in M$. It is easy to check that $P$ is then a solution to $\mathscr{S}$. Thus if $M$ does not immediately give us solution to $\mathscr{S}$ it must be the case that are some number of edges $K > 0$ with $(g_i, L_{b_j}) \in M$ but $(L_{g_i}, b_j) \notin M$, or vice versa. Call these $K$ edges the \textit{mismatched edges}. Under these circumstances we show how to alter matching $M$ to get a new matching $M'$, still of size  $|G_L| + |B_L|$ but with fewer than $K$ mismatched edges. We can run this process at most $K$ times in succession to get a matching without mismatched edges, and then turn M into a solution to $\mathscr{S}$ completing the proof.

Thus, without loss of generality, suppose we have $(g_1, L_{b_1}) \in M$, establishing that $g_1$ and $b_1$ are list-compatible, but $(L_{g_1}, b_1) \notin M$. If there is no edge in $M$ that includes $L_{g_1}$ then we can just replace whatever edge includes $b_1$ in $M$ with $(L_{g_1}, b_1)$ and still have a maximum matching, but with fewer mismatched edges, which would complete the proof. Thus assume we have $(L_{g_1}, b_2) \in M$, with $b_2 \neq b_1$, establishing the list-compatibility of $g_1$ and $b_2$. Now suppose there were no edge in $M$ containing $L_{b_2}$. Then we could swap $(g_1, L_{b_1})$ in $M$ for $(g_1, L_{b_2})$, keep the maximum matching and again reduce the number of mismatched edges, completing the argument.  Thus suppose we have $(g_2, L_{b_2}) \in M$.  As earlier, we must have some edge connected to $L_{g_2}$. If it is $b_1$ then we can replace $(L_{g_1}, b_2), (L_{g_2},b_1)$ with $(L_{g_1}, b_1), (L_{g_2},b_2)$, reducing the number of mismatched edges. Hence we must have $b_3 \notin \{b_1, b_2\}$ with $(L_{g_2}, b_3) \in M$. Continuing in this fashion, given that $|B_L|$ is finite, we eventually must have $(L_{g_j}, b_1) \in M$ for some $j$. The situation is depicted in Figure \ref{fig:mismatched_edges}.
\begin{figure}[h] 
\centerline{\scalebox{0.33}{\includegraphics{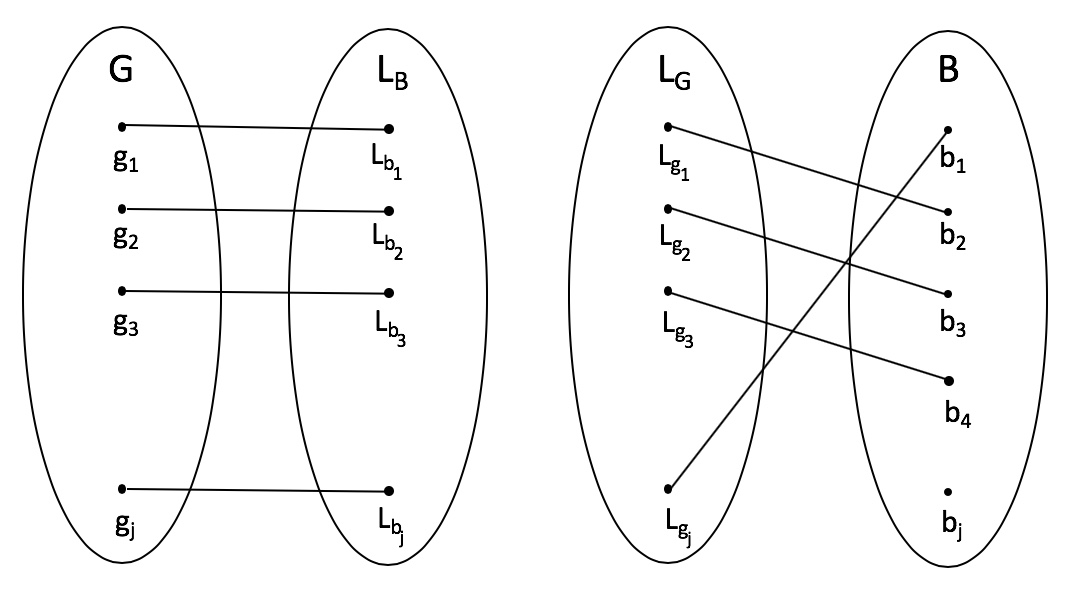}}}
\caption{Mismatched edges of the form $(g_i, L_{b_i}), (L_{g_i}, b_{i+1})$, together with $(g_j, L_{b_j}), (L_{g_j}, b_1)$, as in the proof of Theorem \ref{thm:smp}}. \label{fig:mismatched_edges}
\end{figure} 
At this point we can replace $\{(L_{g_i}, b_{i+1})\}_{i = 1}^{j-1}\} \cup \{(L_{g_j}, b_1)\}$ with $\{(L_{g_i},b_i)\}_{i=1}^j$, retaining a maximum matching and reducing the number of mismatched edges. In all cases, therefore, we see that we can convert $M$ to another maximal matching $M'$, while reducing the number of mismatched edges. The proof is therefore complete.
\end{proof}

Consider the SMP $\mathscr{S} = (G, G_L, B, B_L, \{\mathscr{T}_g\}_{g \in G_L}, \{\mathscr{T}_b\}_{b \in B_L})$. Pare down the list of each $g \in G_L$ such that $b \in \mathscr{T}_g \cap B_L$ only if $g \in \mathscr{T}_b$, in other words, only if $b$ and $g$ are list-compatible. There is no loss in generality in doing so, since if $g \in G_L$ and $b \in \mathscr{T}_g \cap B_L$ then we can only pair $g$ with $b$ if the two elements are list-compatible. Refer to the pared down set of lists as $\{\mathscr{T}^*_g\}_{g \in G_L}$. Analogously, pare down the list of each $b \in B_L$ such that  $g \in \mathscr{T}_b \cap G_L$ only if $g$ and $b$ are list compatible and refer to the pared down set of lists as $\{\mathscr{T}^*_b\}_{b \in B_L}$.

\begin{cor} Consider the SMP $\mathscr{S} = (G, G_L, B, B_L, \{\mathscr{T}_g\}_{g \in G_L}, \{\mathscr{T}_b\}_{b \in B_L})$. Then $\mathscr{S}$ is solvable iff the respective CMPs $\mathscr{C}_G = (G_L, B, \{\mathscr{T}^*_g\}_{g \in G_L})$ and $\mathscr{C}_B = (B_L, G, \{\mathscr{T}^*_b\}_{b \in B_L})$ are both solvable.
\end{cor}

\begin{proof}
If $\mathscr{S}$ is solvable then by definition $\mathscr{C}_G$ and $\mathscr{C}_B$ are both solvable, so assume that $\mathscr{C}_G$ and $\mathscr{C}_B$ are each solvable. As in the proof of the Symmetric Marriage Theorem (Theorem \ref{thm:smp}) consider the graph $G^*(\mathscr{S})$. By virtue of the fact that $\mathscr{C}_G$ has a solution there is a matching of $G_L$ with $L_B \cup B$ of size $|G_L|$. Analogously, since $\mathscr{C}_B$ has a solution there is a matching of $B_L$ with $L_G \cup G$ of size $|B_L|$. These matchings are pairwise disjoint and, moreover, the edges have no vertices in common. Taken together, therefore, they give rise to a maximum matching in $G^*(\mathscr{S})$ of size $|L_G| + |L_B|$. As shown in the proof of the prior theorem, a matching of size $|L_G| + |L_B|$ in $G^*(\mathscr{S})$ guarantees a solution to $\mathscr{S}$, and so the result is established.
\end{proof}

\begin{observ} \label{obs3} We may optionally have $\bigvee$ clauses that contain arbitrary subsets of the variables appearing in any one of the $\oplus_{\leq 1}$ statements, not necessarily all of the variables. \end{observ}
\noindent \textbf{Reason:} In the canonical representation (\ref{eqn:canonical_embedding}), the $\oplus_{\leq 1}$ statements corresponded to the fact that a given boy could be paired with at most one of the girls. A $\bigvee$ clause  with a subset of these variables corresponds to a boy indicating which, of the girls that are interested in him, he is willing to be paired with in an associated SMP, which is solvable as per Theorem \ref{thm:smp}. The case where there are multiple $\bigvee$ clauses of this sort, with a subset of variables coming from the same $\oplus_{\leq 1}$ statement, is covered in the proof of the following theorem, which summarizes all of the observations.

\begin{theorem} \label{thm:sum_up}
Consider any SAT formula that can be written in CNF with clauses strictly of the form $\oplus_{\leq 1}(A_{i1},...,A_{iL_i})$ and $\bigvee(A_{j1},...,A_{jM_j})$, where the $A_{ij}$ appearing across all the $\oplus_{\leq 1}$ clauses are unique and the $A_{kl}$ appearing across the $\bigvee$ clauses are unique except for some $\bigvee$ clauses, each of which has variables that are a subset of the variables appearing in some one $\oplus_{\leq 1}$ statement. Such SAT problems are solvable in polynomial time by translation to the SMP.
\end{theorem}

\begin{proof}
From the earlier observations we know that we can turn any SAT formula with clauses of the form $\oplus_{\leq 1}(A_{i1},...,A_{iL_i})$ and $\bigvee(A_{j1},...,A_{jM_j})$ where the $A_{ij}$ appearing across all the $\oplus_{\leq 1}$ clauses are unique, and the $A_{kl}$ appearing across the $\bigvee$ clauses are unique, into a CMP, and hence into an SMP.  Suppose that in addition we have $\bigvee$ clauses, each of which has variables that come from some one $\oplus_{\leq 1}$ statement. The case in which there is just one such $\bigvee$ clause per $\oplus_{\leq 1}$ statement is covered by Observation \ref{obs3}. Suppose then that there are two $\bigvee$ clauses, each containing a different subset of variables from the same $\oplus_{\leq 1}$ statement. The two $\bigvee$ clauses can plainly not be variable-wise disjoint since otherwise at least two of the variables must be TRUE, contradicting the $\leq 1$ aspect of the associated $\oplus_{\leq 1}$ clause. Thus they must share some number of variables. Suppose the two $\bigvee$ clauses are then of the following form:
\begin{eqnarray}
	& A_1 \vee  A_k \vee B \notag \\
	& A_1 \vee A_k \vee C \label{eqn:latter},
\end{eqnarray}
where the subformulas $B$ and $C$ are some further pairwise disjoint disjuncts of variables coming from the same $\oplus_{\leq 1}$ statement, with $B$ and $C$ not both empty. These two clauses taken together (i.e. conjuncted) imply $(A_1 \vee  A_k) \vee (B \land C)$. However, we can't have both $B$ and $C$ unless one of $B$ and $C$ is empty. But then the two clauses (\ref{eqn:latter}) are equivalent to just $A_1 \vee  A_k$. Thus any two $\bigvee$ clauses with variables coming from the same $\oplus_{\leq 1}$ statement can be collapsed to one such clause and the same is true for an arbitrary number of $\bigvee$ clauses with such shared variables. The theorem follows.
\end{proof}

Just like in Horn-SAT, the issue of rewritability comes up in CMP-SAT.  It is, however, rather easy to see that one can recognize whether a given SAT formual is \textbf{CMP-rewritable} simply by focusing on the $\oplus_{\leq 1}$ statements, which must appear in all-pairs families of length $2$, and their relationship to the other clauses. We therefore omit the proof that one can test whether a SAT formula is CMP-rewritable simply by flipping the polarity of some subset of the variables. A more difficult problem awaits us.

\subsection{Recognition of CMP-SAT from a Non-Canonical Representation}

What if we were given this same CMP instance \textit{not} using the canonical representation but, say, in the messy way (\ref{eqn:messy_3cnf}) in which we saw it the first time around. Would we be able to recognize it as a CMP? So far we have just leveraged the fact that SAT is NP-Complete and we can map any CMP back to a SAT problem. But the FMP is also NP-Complete so we can take any SAT instance and map it to an FMP instance and work with it there.  Ideally we would like to map it back to the FMP, see all integral tradeoffs, and declare victory. However, things are not quite that simple. 
In Appendix \ref{app:rfragment_logic_and_fragment_reductions} we go through this mapping-back-to-the-FMP process and describe a set of rules for simplifying FMPs with the aim of arriving at an all integral FMP, and hence a CMP. We call the rules \textbf{Fragment Logic} and the process of trying to obtain a CMP from an FMP by successive application of the rules \textbf{Fragment Reduction}. Among other things we show that Fragment Reduction can always be done in time and space that is polynomial in the size of the input.

\section{The CMP-FMP Dividing Line}

We obviously can't reduce every FMP to a CMP since the FMP is NP-Complete but the CMP is not. It is natural to wonder at what level of expressivity the FMP flips and becomes NP-Complete. The following result answers a question along those lines posed by Joe Mitchell \cite{MITCHELL:2019}. 
It also provides a second, simpler, NP-Complete reduction for the FMP.

\begin{theorem} The FMP is NP Complete even if all fractions are restricted to be $\frac{1}{2}$.
\end{theorem}

\begin{proof}
The reduction is from 3D Tripartite Matching. Let the three sets of the tripartition be $G,X$ and $Y$. We need to be able to represent triples within the FMP and also guarantee that any solution found is necessarily one in which we pick unique elements from each of the sets. 

$G$ will be the same ground set as in the CMP. Instead of specifying fractions of elements $p_i b_i$ with all $b_i \in B$ we will specify fractions of elements $p_i x_i, q_j y_j: x_i \in X, y_j \in Y$. In the CMP, ordered pairs from the bipartition $G \times B$ were represented by singletons of the form $\{b_i\} \in \mathscr{T}_g$, meaning that $(g, b_i) \in G \times B$. We will represent $3$-tuples by $\{\frac{1}{2}x_i, \frac{1}{2} y_j\} \in \mathscr{T}_g$, meaning that $(g,x_i,y_j) \in G \times X \times Y$. Just like any CMP is a collection of tradeoff sets, each comprising integral tradeoffs, a Tripartite Matching  instance is a collection of tradeoff sets, all of whose tradeoffs are half-integral, with one element coming from $X$ and one from $Y$.

Now let $\{x_i\}_{i=1}^K$ be the collection of elements of $X$ used in all such tradeoffs, and analogously, let $\{y_i\}_{i=1}^L$ be the collection of elements of $Y$ used in the various tradeoffs. Add an additional collection of elements $\{z_i\}_{i=1}^{K+L}$ and to ensure that we don't satisfy the same $x_i$ or $y_j$ in two of the previously listed half-integral tradeoffs, for each $x_i$ or $y_j$ we have previously used, add an additional tradeoff set with the single tradeoff $\{\frac{1}{2}x_i, \frac{1}{2}z_i\}$ or  $\{\frac{1}{2}y_j, \frac{1}{2}z_{K+j}\}$. Being tradeoff sets comprising just a single tradeoff all of these tradeoffs must be satisfied. Hence at most one tradeoff including any of the original elements can be satisfied. It follows that a perfect tripartite matching can be found iff the associated FMP can be solved and so the theorem is established.
\end{proof}

\section{Suggested Further Problems}

Our investigations have led us to many additional questions, an assortment of which we now pose:
\begin{enumerate}[1.]
\item If an FMP is ``really a CMP'' is it always possible to apply the fragment logic rules to achieve the reduction? (The notion of ``really a CMP'' may need to be defined more formally.) Sometimes the fragment logic rules solve part or all of the FMP. If we are able to get to an all integral FMP, how do we know whether we have successfully reduced the problem to a CMP or have just solved the FMP?
\item Are there special cases of the FMP other than the CMP that have polynomial time
solutions? Of course there must be -- what do 2-SAT, Horn-SAT and XOR-SAT look like when mapped onto the FMP and how far do we get applying Fragment Logic to these problems? 
\item Tripartite matching generalizes to more generic 3D-matching, which is therefore also NP-complete. Thus polynomial-time perfect matching in general graphs can be extended to an NP-complete problem and to a bi-directional SAT embedding. What does this associated subset of SAT look like? 
\item In theory what we have done for the CMP can be done for any problem in P. In this way we can dice up the landscape of any NP-Complete problem into a whole spectrum of polynomially solvable problem instances. Is there a way to systematically map out this space? 
\item Are constant factor approximations possible for the FMP? In other words, unless P = NP, we can't say for sure whether $B$ will suffice to solve a given $\mathscr{F} = (G, B, \{\mathscr{T}_g\}_{g \in G})$, but perhaps we can do so for $kB$ and some sufficiently large $k$. 
\end{enumerate}

In Observation \ref{obs3} and Theorem \ref{thm:sum_up} we considered just the simplest class of SMP embeddings into SAT where  the elements of $B_L$ have lists comprised solely of elements of $G_L$, and, furthermore, $B \setminus B_L = \varnothing$. General SMP instances have more complicated SAT embeddnings, though still requiring just the $\bigvee()$ and $\oplus_{\leq 1}()$ operators that characterized the CMP embeddings. A key difference is that in the SMP two $\bigvee()$ statements may share some subset of their variables. We intend to explore this embedding more fully in future work. 

\section*{Acknowledgements}
The author gratefully acknowledges helpful conversations with Estie Arkin, Joe Mitchell and Ken Regan.

\bibliographystyle{plain}
\bibliography{fmp_focs}

\section{Appendix}

\subsection{Proofs of the Normal Form Theorems} \label{app:normal_form_proofs}

 \begin{proof} [Proof of Theorem \ref{thm:fmp_1st_normal_form}]
Suppose we have an FMP $\mathscr{F} = (G, B, \{\mathscr{T}_g\}_{g \in G})$ with $B = \{b_1,...,b_M\}$. Let us first establish that we can replace $\mathscr{F}$ with an equi-satisfiable FMP $\mathscr{F}' = (G', B', \{\mathscr{T}_g\}_{g \in G'})$ with at most three tradeoffs per tradeoff set. Let
\begin{equation*} 
	\mathscr{T}_g = \{t_1, ..., t_N\}
\end{equation*}
be a sample tradeoff set with $N \geq 4$. Satisfying $\mathscr{T}_g$ with $B$ is equivalent to satisfying both
\begin{eqnarray*}
	\mathscr{T}_{g_1}  =~ & \{t_1, ..., t_{\lfloor \frac{N}{2} \rfloor},\{b_{M+1}\}\},  \\
	\mathscr{T}_{g_2}  =~ & \{t_{\lfloor \frac{N}{2} \rfloor + 1}, ..., t_N,\{b_{M+1}\}\}
\end{eqnarray*}
with $B \cup \{b_{M+1}\}$. Since $\mathscr{T}_{g_1}$ and $\mathscr{T}_{g_2}$ each contain strictly fewer tradeoffs than $\mathscr{T}_g$ reducibility to an FMP with at most three tradeoffs per tradeoff set is established.

\medskip

To see that we can arrange to put each fractional tradeoff in a tradeoff set with just one other integral tradeoff, note the tradeoff
\begin{equation*} 
	\mathscr{T} = \{t_1, ..., t_N\}.
\end{equation*}
is equivalent to the family:
\begin{flalign*} \label{eqn:1st_splitting}
	\mathscr{T}_1  =~ & \{t_1,  \{b_{M+1}\}\} \notag \\
	... & \\
	\mathscr{T}_N  =~ & \{t_N,  \{b_{M+N}\}\}  \notag \\
	\mathscr{T}_{N+1} =~ & \{ \{b_{M+1}\}, ..., \{b_{M+N}\}\}, \notag
\end{flalign*}
after which we just split the last tradeoff set into tradeoff sets with no more than $3$ tradeoffs each.\\

The run time for this reduction is $O(KN)$ where $K$ is the maximum number of tradoffs in any tradeoff set and $N = |G|$.
\end{proof}

\begin{proof} [Proof of Theorem \ref{thm:fmp_2nd_normal_form}]
Theorem \ref{thm:fmp_1st_normal_form} provides us with a reduction such that all tradeoff sets containing fractional tradeoffs are of the form
\begin{equation} \label{eqn:standard_tradeoff}
	\mathscr{T} = \{\{p_1 b_1,...,p_N b_N\}, \{b\}\}.
\end{equation}
Furthermore, the proof of Theorem \ref{thm:fmp_1st_normal_form} shows that the element $b$ in the second tradeoff above does not appear \textit{fractionally} in any of the other tradeoffs and tradeoff sets.
But now, satisfying tradeoff set (\ref{eqn:standard_tradeoff}) is equivalent to satisfying the following N tradeoff sets:
\begin{flalign} \label{eqnarray:reduction_to_binaries_prelim_step}
	\mathscr{T}_1 =~ & \{\{p_1 b_1, (1-p_1) b_{N+1}\}, \{\frac{1}{N} b, \frac{N-1}{N} b_{N+2}\}\} \notag\\
	\mathscr{T}_2 =~ & \{\{p_2 b_2, (1-p_2) b_{N+3}\}, \{\frac{1}{N} b, \frac{N-1}{N} b_{N+4}\}\} \notag \\
	... \\
	\mathscr{T}_N =~ & \{\{p_N b_N, (1-p_N) b_{3N-1)}\}, \{\frac{1}{N} b, \frac{N-1}{N} b_{3N}\}\} \notag
\end{flalign}
Note that the second fragments in all of these tradeoffs are free. The collection of tradeoff sets $\mathscr{T}_1,...,\mathscr{T}_N$ can now be further reduced to a collection such that there is only one fractional tradeoff per tradeoff set using Theorem \ref{thm:fmp_1st_normal_form}. We note that the fact that the second fragments in all the tradeoff sets are free is preserved. The runtime for both processes is $O(F)$ where $F$ is the total number of fragments appearing across all tradeoffs and tradeoff sets. The proof of the present theorem is thus complete. 
\end{proof}

\subsection{Verification of Non-Horn-rewritability} \label{app:non_horn_rewritability}

\begin{proof} [Proof that $\phi$ is not Horn-rewritable in Theorem \ref{thm:no_union}]
To aid in the following of our argument we refer to Table \ref{table:clauses}, which gives a listing of positive and negative literals by clause number in each of the sub-formulas $\phi_1, \phi_2, \phi_3$ and $\phi_4$ from (\ref{eqn:messy_3cnf}) in the main body of the paper.

\begin{table}[ht] 
\centering
\begin{small}
\begin{tabular}{| c | c | c | c | c | c | c | c | c | c | c |}
\hline
 & \textbf{Clause \#} & $\boldsymbol{b_{11}}$ & $\boldsymbol{b_{12}}$ & $\boldsymbol{b_{14}}$ & $\boldsymbol{b_{21}}$ & $\boldsymbol{b_{23}}$ & $\boldsymbol{b_{31}}$ & $\boldsymbol{b_{32}}$ & $\boldsymbol{b_{43}}$ & $\boldsymbol{b_{44}}$ \\
\hline
$\boldsymbol{\phi_1}$ & 1 & P & & & P & & P & & & \\ 
\hline
$\boldsymbol{\phi_1}$ & 2 & & N & & P & & P & & & \\ 
\hline
$\boldsymbol{\phi_1}$ & 3 & & & N & P & & P & & & \\ 
\hline
$\boldsymbol{\phi_1}$ & 4 & P & & & & N & P & & & \\ 
\hline
$\boldsymbol{\phi_1}$ & 5 & & N & & & N & P & & & \\ 
\hline
$\boldsymbol{\phi_1}$ & 6 & & & N & & N & P & & & \\ 
\hline
$\boldsymbol{\phi_1}$ & 7 & P & & & P & & & N & & \\ 
\hline
$\boldsymbol{\phi_1}$ & 8 & & N & & P & & & N & & \\ 
\hline
$\boldsymbol{\phi_1}$ & 9 & & & N & P & & & N & & \\ 
\hline
$\boldsymbol{\phi_1}$ & 10 & P & & & & N & & N & & \\ 
\hline
$\boldsymbol{\phi_1}$ & 11 & & N & & & N & & N & & \\ 
\hline
$\boldsymbol{\phi_1}$ & 12 & & & N & & N & & N & & \\ 
\hline
$\boldsymbol{\phi_2}$ & 1 & N & & & & & N & & & \\ 
\hline
$\boldsymbol{\phi_2}$ & 2 & & P & & & & N & & & \\ 
\hline
$\boldsymbol{\phi_2}$ & 3 & & & N & & & N & & & \\ 
\hline
$\boldsymbol{\phi_2}$ & 4 & N & & & & & & P & & \\ 
\hline
$\boldsymbol{\phi_2}$ & 5 & & P & & & & & P & & \\ 
\hline
$\boldsymbol{\phi_2}$ & 6 & & & N & & & & P & & \\ 
\hline
$\boldsymbol{\phi_3}$ & 1 & & & & N & & & & P & \\ 
\hline
$\boldsymbol{\phi_3}$ & 2 & & & & & P & & & P & \\ 
\hline
$\boldsymbol{\phi_3}$ & 3 & & & & N & & & & & N\\ 
\hline
$\boldsymbol{\phi_3}$ & 4 & & & & & P & & & & N\\ 
\hline
$\boldsymbol{\phi_4}$ & 1 & N & & & & & & & N & \\ 
\hline
$\boldsymbol{\phi_4}$ & 2 & & N & & & & & & N & \\ 
\hline
$\boldsymbol{\phi_4}$ & 3 & & & P & & & & & N & \\ 
\hline
$\boldsymbol{\phi_4}$ & 4 & N & & & & & & & & P\\ 
\hline
$\boldsymbol{\phi_4}$ & 5 & & N & & & & & & & P\\ 
\hline
$\boldsymbol{\phi_4}$ & 6 & & & P & & & & & & P\\ 
\hline
\end{tabular}
\end{small}
\caption{Listing of positive and negative literals in $\phi$ by clause number: For each of the sub-formulas $\phi_1, \phi_2, \phi_3$ and $\phi_4$ we list which of the variables $b_{ij}$ appear in the respective clauses. The clauses are indexed sequentially. ``P'' indicates that a positive instance of the variable appears and ``N'' indicates that a negative instance of the variable appears.} \label{table:clauses}
\end{table}

Let $\phi_i{:}j$ denote the $j$th clause of $\phi_i$. Further, let $b_{ij} \hookrightarrow \neg b_{ij}$ denote the decision to swap the parity of the variable $b_{ij}$ everywhere in $\phi$, and so send $b_{ij}$ to $\neg b_{ij}$ and $\neg b_{ij}$ to $b_{ij}$.  

Let us assume that $\phi$ is Horn-rewritable. By $\phi_1{:}7$, we must have either $b_{11} \hookrightarrow \neg b_{11}$ or $b_{21} \hookrightarrow \neg b_{21}$. Suppose that $b_{11} \hookrightarrow \neg b_{11}$. By  $\phi_2{:}4$, $b_{32} \hookrightarrow \neg b_{32}$, and then, by either $\phi_1{:}8$ or $\phi_1{:}9$, $b_{21} \hookrightarrow \neg b_{21}$. It follows that we must have $b_{21} \hookrightarrow \neg b_{21}$. 

It follows then from $\phi_3{:}1$ that $b_{43} \hookrightarrow \neg b_{43}$, and then from $\phi_4{:}3$, that $b_{14} \hookrightarrow \neg b_{14}$. Finally, from $\phi_2{:}6$, $b_{32} \hookrightarrow \neg b_{32}$. However, this means we have a problem in clause $\phi_1{:}9$ or $\phi_1{:}12$ in that we have committed to both $b_{14} \hookrightarrow \neg b_{14}$ and $b_{32} \hookrightarrow \neg b_{32}$, yielding positive $b_{14}$ and $b_{32}$ literals in both clauses, and hence the formula $\phi$ is not Horn-rewritable. 
\end{proof}

\subsection{Example of Recognizing and Solving a CMP-SAT Instance in its Canonical Form} \label{app:recognizing_and_solving}

To illustrate the process of recognizing and solving a CMP-SAT instance in its canonical presentation, consider the following example:
\begin{flalign*}
	& (A \vee B \vee X) \land (\neg X \vee C \vee Y) \land (\neg Y \vee D \vee E) \\
	& F \vee G \\
	& H \vee I \vee J \\
	& (\neg I \vee \neg F) \land (\neg D \vee \neg I) \land (\neg D \vee \neg F) \\
	& (\neg C \vee \neg H) \land (\neg G \vee \neg C) \land (\neg H \vee \neg G) \\
	& \neg E \vee \neg J  
\end{flalign*}
First remove the bridge variables from first sub-formula to create long conjunction of positive literals. Then form $\oplus_{\leq 1}$ clauses from the last three sub-formulas:
\begin{flalign*}
	& A \vee B \vee C \vee D \vee E \\
	& F \vee G \\
	& H \vee I \vee J \\
	& \oplus_{\leq 1}(D, F, I) \\
	& \oplus_{\leq 1}(C, G, H) \\
	& \oplus_{\leq 1}(E,J)
\end{flalign*}

Next we turn the top three sub-formulas into tradeoff sets and pick representatives for each of the $\oplus_{\leq 1}$ statements:
\begin{flalign*}
	& \mathscr{T}_1 = \{A, B, C, D, E\} \\
	& \mathscr{T}_2 = \{F, G\} \\
	& \mathscr{T}_3 = \{H, I, J\}  \\
	& \oplus_{\leq 1}(D, F, I)~~\textrm{Rep:  }R_1\\
	& \oplus_{\leq 1}(C, G, H) ~~\textrm{Rep:  }R_2\\
	& \oplus_{\leq 1}(E,J)~~\textrm{Rep:  }R_3\\
\end{flalign*}
Next we replace any variables in the tradeoff sets with their representatives from the $\oplus_{\leq 1}$ statements and make a note of the mappings:
\begin{flalign*}
	& \mathscr{T}_1 = \{A, B, R_2, R_1, R_3\}~~C \hookrightarrow R_2, D \hookrightarrow R_1,  E \hookrightarrow R_3\\
	& \mathscr{T}_2 = \{R_1, R_2\}~~F \hookrightarrow R_1, G \hookrightarrow R_2\\
	& \mathscr{T}_3 = \{R_2, R_1, R_3\}~~H \hookrightarrow R_2, I \hookrightarrow R_1,  J \hookrightarrow R_3 
\end{flalign*}
Finally arriving at an easily solvable CMP. One solution is:  $T(\mathscr{T}_1) = A, T(\mathscr{T}_2) = R_1,$\\$T(\mathscr{T}_3) = R_2$, which in turn means we set the original variables $A = F = H = $ TRUE and all other \textbf{non-bridge} variables to FALSE to satisfy the associated SAT formula.

Setting the original variables $A = F = H = $ TRUE and all other non-bridge variables to FALSE we satisfy all of the associated SAT sub-formulas except the first one, which contains bridge variables:
\begin{flalign*}
	& (A \vee B \vee X) \land (\neg X \vee C \vee Y) \land (\neg Y \vee D \vee E) \\
	& F \vee G \\
	& H \vee I \vee J \\
	& (\neg I \vee \neg F) \land (\neg D \vee \neg I) \land (\neg D \vee \neg F) \\
	& (\neg C \vee \neg H) \land (\neg G \vee \neg C) \land (\neg H \vee \neg G) \\
	& \neg E \vee \neg J  
\end{flalign*}
To satisfy the first sub-formula, given that, of the primary variables $A, B, C, D, E$, only $A$ is TRUE, we must additionally set the bridge variable $Y=$ FALSE, and then also $X=$ FALSE.
 
 \subsection{Fragment Logic and Fragment Reductions} \label{app:rfragment_logic_and_fragment_reductions}
 
In this section we describe the process of mapping from a SAT problem to a corresponding FMP problem, and then working within the context of the FMP to see if we actually might have a CMP.  We illustrate the process by working backwards from the canonical embedding of our familiar problem (\ref{eqn:canonical_embedding2}). On the left below are the clauses of the canonical embedding, and on the right, the associated mapping to the FMP introduced in our proof of the NP-completeness of the FMP. Elements $b_\infty,...,b_{\infty + 11}$ are the auxillary elements introduced in that proof.
\begin{table} [h]
\begin{tabular}{ l | l }
$b_{11} \vee b_{21} \vee b_{31}$ & $\{\frac{1}{3}b_{11}, \frac{2}{3}b_{\infty}\}, \{\frac{1}{2}b_{21},\frac{1}{2}b_{\infty+1}\},\{\frac{1}{2}b_{31}, \frac{1}{2}b_{\infty+2}\}$ \\
$b_{12} \vee b_{32}$  & $\{\frac{1}{3}b_{12}, \frac{2}{3}b_{\infty+3}\}, \{\frac{1}{2}b_{32},\frac{1}{2}b_{\infty+4}\}$ \\ 
$b_{23} \vee b_{43}$ & $\{\frac{1}{2}b_{23}, \frac{1}{2}b_{\infty+5}\}, \{\frac{1}{2}b_{43},\frac{1}{2}b_{\infty+6}\}$ \\ 
$b_{14} \vee b_{44}$ & $\{\frac{1}{3}b_{14}, \frac{2}{3}b_{\infty+7}\}, \{\frac{1}{2}b_{44},\frac{1}{2}b_{\infty+8}\}$ \\ 
$\neg b_{11} \vee \neg b_{12}$ & $\{\frac{1}{2}b_{\infty}, \frac{1}{2}b_{\infty+9}\}, \{\frac{1}{2}b_{\infty+3}, \frac{1}{2}b_{\infty+10}\}$ \\
$\neg b_{11} \vee \neg b_{14}$ & $\{\frac{1}{2}b_{\infty}, \frac{1}{2}b_{\infty+9}\}, \{\frac{1}{2}b_{\infty+7}, \frac{1}{2}b_{\infty+11}\}$ \\
$\neg b_{12} \vee \neg b_{14}$ & $\{\frac{1}{2}b_{\infty+3}, \frac{1}{2}b_{\infty+10}\}, \{\frac{1}{2}b_{\infty+7}, \frac{1}{2}b_{\infty+11}\}$ \\
$\neg b_{21} \vee \neg b_{23}$ & $\{b_{\infty+1}\}, \{b_{\infty+5}\}$\\ 
$\neg b_{31} \vee \neg b_{32}$ & $\{b_{\infty+2}\}, \{b_{\infty+4}\}$\\ 
$\neg b_{43} \vee \neg b_{44}$ & $\{b_{\infty+6}\}, \{b_{\infty+8}\}$ 
\end{tabular}
\end{table}

Over the course of working through numerous examples we have distilled a set of simplification principles that we call the rules of \textbf{Fragment Logic}. 
Each application of Fragment Logic either involves the deletion/removal of fragments, or, if fragments are not removed, then the number of fragments is kept constant and either the number of integral fragments (fragments $pe$ with $p = 1$) is increased, the number of non-integral fragments is decreased, or both. Thus only a polynomial number of these operations can ever be performed. We refer to the process of iteratively applying Fragment Logic as \textbf{Fragment Reduction}. 
The first rule of Fragment Logic is the following:
\begin{enumerate} [R1.]
\item [R1.] One may eliminate any element for which the sum of all associated fragments is $\leq 1$. Note that this may leave some tradeoffs with fragments whose fractions sum to less than $1$.
\end{enumerate}

Applying this rule we obtain:\\

\smallskip
\begin{tabular}{ l | l }
$b_{11} \vee b_{21} \vee b_{31}$ & $\{\frac{2}{3}b_{\infty}\}, \{\frac{1}{2}b_{\infty+1}\},\{\frac{1}{2}b_{\infty+2}\}$ \\
$b_{12} \vee b_{32}$  & $\{\frac{2}{3}b_{\infty+3}\}, \{\frac{1}{2}b_{\infty+4}\}$ \\ 
$b_{23} \vee b_{43}$ & $\{\frac{1}{2}b_{\infty+5}\}, \{\frac{1}{2}b_{\infty+6}\}$ \\ 
$b_{14} \vee b_{44}$ & $\{\frac{2}{3}b_{\infty+7}\}, \{\frac{1}{2}b_{\infty+8}\}$ \\ 
$\neg b_{11} \vee \neg b_{12}$ & $\{\frac{1}{2}b_{\infty}\}, \{\frac{1}{2}b_{\infty+3}\}$ \\
$\neg b_{11} \vee \neg b_{14}$ & $\{\frac{1}{2}b_{\infty}\}, \{\frac{1}{2}b_{\infty+7}\}$ \\
$\neg b_{12} \vee \neg b_{14}$ & $\{\frac{1}{2}b_{\infty+3}\}, \{\frac{1}{2}b_{\infty+7}\}$ \\
$\neg b_{21} \vee \neg b_{23}$ & $\{b_{\infty+1}\}, \{b_{\infty+5}\}$\\ 
$\neg b_{31} \vee \neg b_{32}$ & $\{b_{\infty+2}\}, \{b_{\infty+4}\}$\\ 
$\neg b_{43} \vee \neg b_{44}$ & $\{b_{\infty+6}\}, \{b_{\infty+8}\}$ 
\end{tabular}

\smallskip

The next rule is the analog of recognizing and simplifying the ``at most one of'' relation in the context of the FMP:
\begin{enumerate} [R1.]
\item [R2.]Suppose we have a collection of tradeoffs $t_1,...,t_k$ that together use a set of elements $e_1,...,e_m$ such that the elements $e_1,...,e_m$ appear in no additional tradeoff beyond these $t_i$. Further, suppose that no two of the tradeoffs $t_1,...,t_k$ can be simultaneously satisfied. Then we may replace each $t_i ~(1 \leq i \leq k)$ by the single element tradeoff $\{e_1\}$.
\end{enumerate}

\noindent In our example we just have very simple cases where, e.g., $\{\frac{1}{2}b_{\infty+1}\}$ in the first tradeoff set, coupled with $\{b_{\infty+1}\}$ in the third-to-last tradeoff set, together allow us to replace the first tradeoff with $\{b_{\infty+1}\}$. However, below is a more sophisticated example 
\begin{figure}[h]
\scalebox{0.2}{\includegraphics{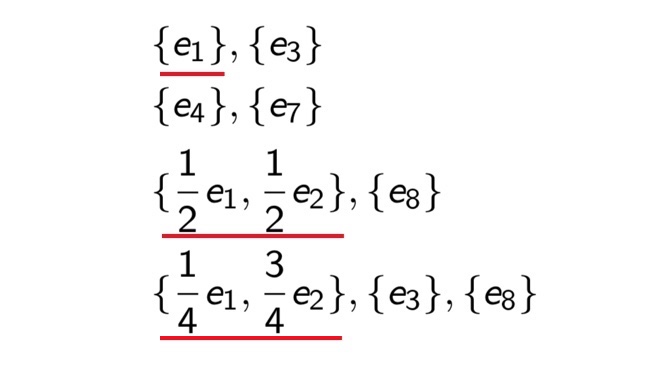}}
\end{figure}
where each of the underlined tradeoffs can be replaced by $\{e_1\}$.  Applying rule 2 to our example, and replacing the $\infty$ subscript with $0$, to assist with the optics, we obtain:
\begin{table} [h]
\begin{tabular}{ l | l }
$b_{11} \vee b_{21} \vee b_{31}$ & $\{\frac{2}{3}b_0\}, \{b_1\},\{b_2\}$ \\
$b_{12} \vee b_{32}$  & $\{\frac{2}{3}b_3\}, \{b_4\}$ \\ 
$b_{23} \vee b_{43}$ & $\{b_5\}, \{b_6\}$ \\ 
$b_{14} \vee b_{44}$ & $\{\frac{2}{3}b_7\}, \{b_8\}$ \\ 
$\neg b_{11} \vee \neg b_{12}$ & $\{\frac{1}{2}b_0\}, \{\frac{1}{2}b_3\}$ \\
$\neg b_{11} \vee \neg b_{14}$ & $\{\frac{1}{2}b_0\}, \{\frac{1}{2}b_7\}$ \\
$\neg b_{12} \vee \neg b_{14}$ & $\{\frac{1}{2}b_3\}, \{\frac{1}{2}b_7\}$ \\
$\neg b_{21} \vee \neg b_{23}$ & $\{b_1\}, \{b_5\}$\\ 
$\neg b_{31} \vee \neg b_{32}$ & $\{b_2\}, \{b_4\}$\\ 
$\neg b_{43} \vee \neg b_{44}$ & $\{b_6\}, \{b_8\}$ 
\end{tabular}
\end{table}

The next rule covers the case where we have an element whose fragments can be partitioned into two classes, one for which the fractions are high and the other for which the fractions are low, and such that we can either use the element to satisfy exactly one of the high fragments or all of the low fragments.

\begin{enumerate} [R1.]
\item [R3.] Let $k, \ell \in \mathbb{Z}^+$ with $\ell \geq 2$. Suppose an element $e$ appears a total of $k + \ell$ times across all tradeoffs and has $k$  fragments for which it appears with fractions $\{p_i\}_{i=1}^k$ such that all $p_i > \frac{1}{2}$, and $\ell$ fragments for which it appears with fractions $\{q_i\}_{i=1}^{\ell}$ such that all $q_i \leq \frac{1}{\ell}$. Further, suppose that $p_i + q_j > 1$ for any $i,j$.  Then we may replace all instances of $p_i e$ with $e$ and all instances of $q_j e$ with $\frac{1}{N}e$ for any $N \geq \ell$. Note that doing so may create tradeoffs with total weight, $\sum p_i > 1$, which we allow.
\end{enumerate}

\noindent Applying rule 3 we obtain:

\smallskip
\begin{tabular}{ l | l }
$b_{11} \vee b_{21} \vee b_{31}$ & $\{b_0\}, \{b_1\},\{b_2\}$ \\
$b_{12} \vee b_{32}$  & $\{b_3\}, \{b_4\}$ \\ 
$b_{23} \vee b_{43}$ & $\{b_5\}, \{b_6\}$ \\ 
$b_{14} \vee b_{44}$ & $\{b_7\}, \{b_8\}$ \\ 
$\neg b_{11} \vee \neg b_{12}$ & $\{\frac{1}{N}b_0\}, \{\frac{1}{N}b_3\}$ \\
$\neg b_{11} \vee \neg b_{14}$ & $\{\frac{1}{N}b_0\}, \{\frac{1}{N}b_7\}$ \\
$\neg b_{12} \vee \neg b_{14}$ & $\{\frac{1}{N}b_3\}, \{\frac{1}{N}b_7\}$ \\
$\neg b_{21} \vee \neg b_{23}$ & $\{b_1\}, \{b_5\}$\\ 
$\neg b_{31} \vee \neg b_{32}$ & $\{b_2\}, \{b_4\}$\\ 
$\neg b_{43} \vee \neg b_{44}$ & $\{b_6\}, \{b_8\}$ 
\end{tabular}

\smallskip

\noindent with the understanding that we may push the fractions $\frac{1}{N}$ as low as we need in any potential assignment of the associated elements $b_0, b_3$ and $b_7$.

\begin{enumerate} [R1.]
\item [R4.] If we have a tradeoff set with exactly two integral tradeoffs $\{\{e_1\},\{e_2\}\}$ and elsewhere the elements $e_1$ and $e_2$ appear solely in integral tradeoffs then we may remove the tradeoff set $\{\{e_1\},\{e_2\}\}$ and wherever $e_2$ appears replace it by $e_1$.
\end{enumerate}

\noindent This rule follows from the fact that we must use $e_1$ or $e_2$ to satisfy $\{\{e_1\},\{e_2\}\}$, and we can use the element we didn't chose exactly once elsewhere.  Applying rule 4, using the pairs of integral tradeoffs in the last three tradeoff sets, we obtain:
\begin{table} [h]
\begin{tabular}{ l | l }
$b_{11} \vee b_{21} \vee b_{31}$ & $\{b_0\}, \{b_1\},\{b_2\}$ \\
$b_{12} \vee b_{32}$  & $\{b_3\}, \{b_2\}$ \\ 
$b_{23} \vee b_{43}$ & $\{b_1\}, \{b_6\}$ \\ 
$b_{14} \vee b_{44}$ & $\{b_7\}, \{b_6\}$ \\ 
$\neg b_{11} \vee \neg b_{12}$ & $\{\frac{1}{N}b_0\}, \{\frac{1}{N}b_3\}$ \\
$\neg b_{11} \vee \neg b_{14}$ & $\{\frac{1}{N}b_0\}, \{\frac{1}{N}b_7\}$ \\
$\neg b_{12} \vee \neg b_{14}$ & $\{\frac{1}{N}b_3\}, \{\frac{1}{N}b_7\}$ \\
$\neg b_{21} \vee \neg b_{23}$ & \\ 
$\neg b_{31} \vee \neg b_{32}$ & \\ 
$\neg b_{43} \vee \neg b_{44}$ & 
\end{tabular}
\end{table}

The next rule amounts to recognizing the encoding of $\oplus_{\leq 1}(e_1,...,e_k)$ as $k \choose{2}$ disjuncts of pairs of negations.
\begin{enumerate} [R1.]
\item [R5.] Suppose we have a collection of elements $\{e_1,...,e_k\}$ and these elements appear in tradeoff sets of the form $\{\{\frac{1}{N} e_i\},\{\frac{1}{N} e_j\}\}$ for some $N \geq 2, \forall{i,j}: 1 \leq i,j \leq k$, and there are no other non-integral fragments containing any of the elements $e_1,...,e_k$. Then we may remove all the tradeoffs with fractional elements and, moreover, replace all appearances of $e_1,...,e_k$ with $e_1$.
\end{enumerate}

\noindent Applying rule 5 we obtain the final simplification, which is just a renamed version of the original CMP: 
\begin{table} [!htbp]
\begin{tabular}{ l | l }
$b_{11} \vee b_{21} \vee b_{31}$ & $\{b_0\}, \{b_1\},\{b_2\}$ \\
$b_{12} \vee b_{32}$  & $\{b_0\}, \{b_2\}$ \\ 
$b_{23} \vee b_{43}$ & $\{b_1\}, \{b_6\}$ \\ 
$b_{14} \vee b_{44}$ & $\{b_0\}, \{b_6\}$ \\ 
$\neg b_{11} \vee \neg b_{12}$ &  \\
$\neg b_{11} \vee \neg b_{14}$ &  \\
$\neg b_{12} \vee \neg b_{14}$ &  \\
$\neg b_{21} \vee \neg b_{23}$ & \\ 
$\neg b_{31} \vee \neg b_{32}$ & \\ 
$\neg b_{43} \vee \neg b_{44}$ & 
\end{tabular}
\end{table}

\medskip

The same process can be applied to the first, less canonical, translation of this problem that we made. The reduction is much more tedious but no less difficult. For this purpose we add two final Fragment Logic rules. The first rule is a convenience rule and can actually be dispensed with (see the proof of Theorem \ref{thm:fragment_reduction_is_polynomial}).

\begin{enumerate} [R1.]
\item [R6.] Suppose we have tradeoff sets $\{\{pe\},t\}$ and $\{\{p'e'\},t'\}$ for tradeoffs $t,t'$, $0 \leq p,p' \leq 1$ and elements $e,e'$. Moreover, suppose that any time either $e$ or $e'$ appears in additional tradeoffs they appear together, i.e. in tradeoffs of the form $\{qe, q'e',...\}$ where the ellipsis ... indicates that there may be other fractional elements in addition to $e, e'$. Further, suppose that in any such tradeoff $p + q > 1$ and $p' + q' > 1$. Then we may replace $\{\{pe\},t\}$ and $\{\{p'e'\},t'\}$ by the single tradeoff set $\{\{pe\}, t \cup t'\}$ where $t \cup t'$ is the tradeoff obtained by taking the union of all fragments in either $t$ or $t'$ and in cases where $pe \in t, p'e \in t'$, as lomg as $p + p\ \leq 1$ we add $(p + p')e$ to $t \cup t'$. If $p + p' > 1$ then $t \cup t'$ cannot be satisfied and so we can remove all the tradeoffs of the form $\{qe, q'e',...\}$ and replace $\{\{pe\},t\}$ and $\{\{p'e'\},t'\}$ by the single tradeoff, tradeoff set $\{\{pe\}\}$. 

\medskip

Analogously, if we have tradeoff sets $\{\{p_1 e_1\},...,\{p_N e_N\},t\}$ and $\{\{p'_1 e'_1\},...,\{p'_N e'_N\},t'\}$, and any time any of the $e_i$ and $e'_i$ appear in the additional tradeoffs they appear together and with fractions $q_i, q'_i$ such that $p_i + q_i > 1$ and $p'_i + q'_i > 1$. Then we may replace $\{\{p_1 e_1\},...,\{p_N e_N\},t\}$ and $\{\{p'_1 e'_1\},...,\{p'_N e'_N\},t'\}$ with the single tradeoff set $\{\{p_1 e_1\},...,\{p_N e_N\},t \cup t'\}$, with the same explanation and proviso about $t \cup t'$ as previously.
\end{enumerate}

\noindent The rationale for this rule, in its first variation, is that either we can satisfy both $\{pe\}$ and $\{p'e'\}$ in the first mentioned tradeoff sets, or we can satisfy some number of the tradeoffs $\{qe, q'e',...\}$ where $e, e'$ appear together, but we cannot satisfy a mixture of these (i.e., we can't satisfy $\{pe\}$ \textit{and} some number of the $\{qe, q'e',...\}$). Thus satisfying just $\{pe\}$ is equivalent, from a satisfiability perspective, to satisfying both $\{pe\}$ and $\{p'e'\}$. Finally, if we satisfy some number of the $\{qe, q'e',...\}$ then we must satisfy both $t$ and $t'$ and hence the construction of $t \cup t'$ as described. The rationale for the second variation is similar.

Here is an example taken from the reduction of the non-canonical representation of the CMP that uses the second variation of this rule: 
\begin{figure}[h]
\centerline{\scalebox{0.25}{\includegraphics{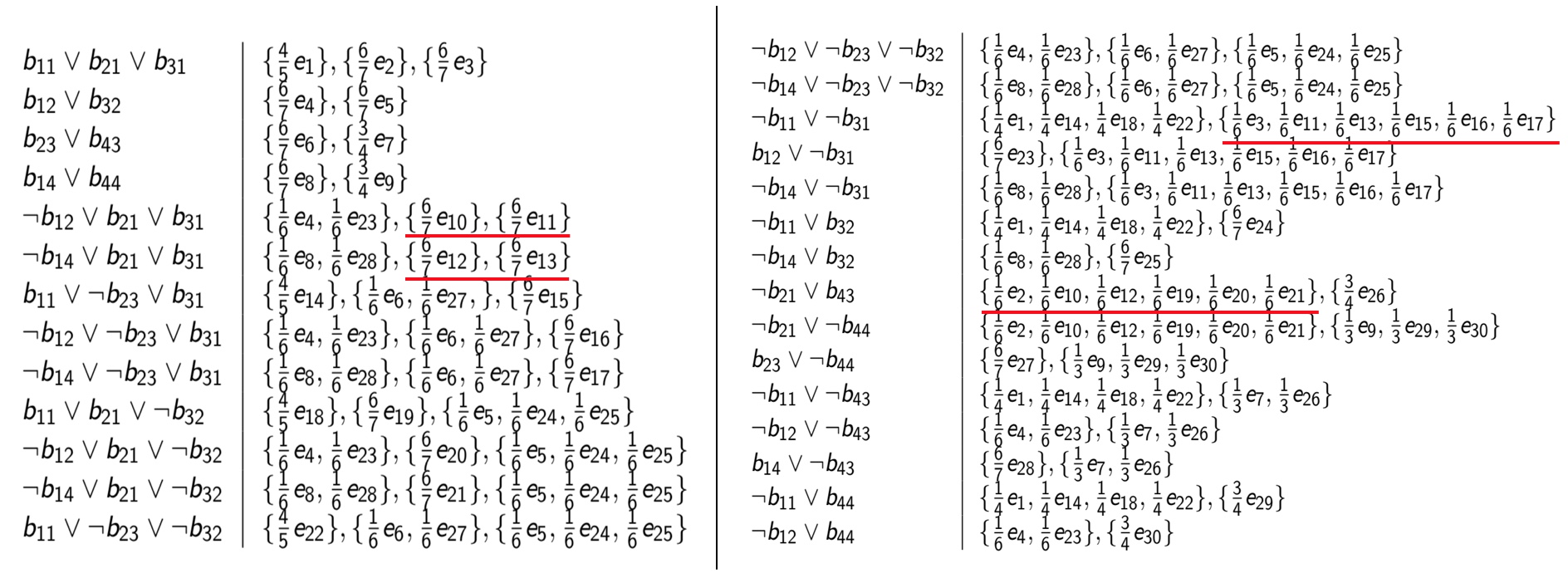}}}
\end{figure}

\noindent In this example, except for the two underlined tradeoffs where $e_{10}$ and $e_{12}$ appear separately, in all other tradeoff sets, they appear together in the same tradeoff (the last underlined tradeoff on the right is one example). The same is true for $e_{11}$ and $e_{13}$, with the first underlined tradeoff on the right serving as an example for that case.  Rule 6 thus applies and we may replace the two tradeoff sets $\{\{\frac{1}{6}e_4,\frac{1}{6}e_{23}\},\{\frac{6}{7}e_{10}\},\{\frac{6}{7}e_{11}\}\}$ and $\{\{\frac{1}{6}e_8,\frac{1}{6}e_{28}\},\{\frac{6}{7}e_{12}\},\{\frac{6}{7}e_{13}\}\}$ with the equi-satisfiable single tradeoff set $\{\{\frac{1}{6}e_4,\frac{1}{6}e_{23},\frac{1}{6}e_8,\frac{1}{6}e_{28}\},\{\frac{6}{7}e_{10}\},\{\frac{6}{7}e_{11}\}\}$.

\begin{enumerate} [R1.]
\item [R7.] Suppose we have tradeoff sets
\begin{eqnarray*}
	& t_1,...,t_k,\{pe,p_1 e_1,...,p_\ell e_\ell\} \\
	& t'_1,...,t'_{k'},\{p'e,p'_1 e'_1,...,p'_{\ell'} e'_{\ell'}\}
\end{eqnarray*}
with $k,k' \geq 1; \ell,\ell' \geq 0$ and $p + p' > 1$, and such that there are no further instances of $e$ or any of the elements $e_1,...,e_\ell$ and $e'_1,...,e'_{\ell'}$. Then the above tradeoff sets are equivalent to the single tradeoff set
\begin{equation*}
	t_1,...,t_k,t'_1,...,t'_{k'}.
\end{equation*}
We note that this last rule is the analog of recognizing bridge variables in a long $\bigvee$ clause.

\end{enumerate}

\begin{theorem} \label{thm:fragment_reduction_is_polynomial}
Fragment Reduction can be run to conclusion, in other words, until no further application of the rules of Fragment Logic is possible, in polynomial time and space.
\end{theorem}

\begin{proof} Key to the proof is to first apply the two normal form reductions, theorems (\ref{thm:fmp_1st_normal_form}) and (\ref{thm:fmp_2nd_normal_form}), leaving fractional tradeoffs containing exactly two fragments, the second of which can be assumed to be free. We then apply rule R1 to remove all such second  elements leaving an FMP with all single element tradeoffs, some of which are fractional and some of which are integral.  Crucially, as we apply each of the rules of fragment logic we will preserve this all-single-element tradeoff (ASET) property.

Clearly R1 can be applied in polynomial time and space, preserving the ASET property. We assume that we apply this rule as needed so that there are never any free fragments to be concerned with.

Next consider rule R2. Given that the ASET property holds, the only way the criteria of the rule can hold is if all tradeoffs contain the same element $e_i$, with associated fragments $\{p_i e_i\}$ such that $p_i + p_j > 1$ for any $p_i, p_j$ with $i \neq j$. This rule can thus clearly be checked and implemented in polynomial time and space without disturbing the ASET property.

It is trivial to check that rules R3 and R4 can be checked and implemented in polynomial time and space, and also preserve the ASET property.

The trick to rule R5 is checking for the $k\choose{2}$ pairs of the form $\{\{\frac{1}{N} e_i\},\{\frac{1}{N} e_j\}\}$. There are some number of pairs of the form $\{\{\frac{1}{N} e_i\},\{\frac{1}{N} e_j\}\}$ in the FMP. By virtue of how the rule is stated if we have a pair $\{\{\frac{1}{N} e_i\},\{\frac{1}{N} e_j\}\}$, then $e_i$ and $e_j$ are either in some all-pairs group of elements or in none. Similarly, if we have a maximal chain, $\{\{\frac{1}{N} e_{i_1}\},\{\frac{1}{N} e_{i_2}\}\},\{\{\frac{1}{N} e_{i_2}\},\{\frac{1}{N} e_{i_3}\}\},...,\{\{\frac{1}{N} e_{i_{k-1}}\},\{\frac{1}{N} e_{i_k}\}\}$, then either the collection of elements $e_{i_1},...,e_{i_k}$ forms an all-pairs group or can be excluded from consideration. It is thus easy to see that we can find a (necessarily pairwise disjoint) full collection of all-pairs groups satisfying the conditions of the rule in polynomial time and space. The rule can also easily be applied in polynomial time and space and preserves the ASET property. 

Regarding rule R6, we note that the condition ``any time either $e$ or $e'$ appear in additional tradeoffs they appear together as $\{qe, q'e',...\}$'' must be vacuous, on account of the fact that we are maintaining the ASET property. Thus the conditions of the rule will render $e$ and $e'$ free and they can then both be removed by rule R1. Thus we see that rule R6 is a convenience rule, and only useful if we have not (as we have) gone through the steps of reducing the FMP via the normal form reductions, applying rule R1, and then preserving the ASET property.

Under the ASET property the conditions of rule R7 reduce to 
\begin{eqnarray*}
	& t_1,...,t_k,\{pe\} \\
	& t'_1,...,t'_{k'},\{p'e\}
\end{eqnarray*}
with $k,k' \geq 1$ and $p + p' > 1$. The recognition and application of the rule can then plainly be done in polynomial time and space while preserving the ASET property.
Since, as noted earlier, one can only sequentially perform a polynomial number of Fragment Reductions, the theorem follows.
\end{proof}

\end{document}